\documentclass[runningheads]{llncs}

\usepackage[british]{babel}
\usepackage[T1]{fontenc}
\usepackage[utf8]{inputenc}
\usepackage{lmodern}
\usepackage[english]{isodate}

\usepackage[ruled,norelsize]{algorithm2e}

\usepackage{amsmath}
\usepackage{amssymb}
\usepackage{bm}

\DeclareMathOperator{\takum}{\tau}

\DeclareMathOperator{\integer}{uint}

\newcommand{\euler}{\sqrt{\mathrm{e}}}

\usepackage{graphicx}
\usepackage{pgf, tikz}
\usepackage{pgfplots}
\pgfplotsset{compat=1.17}
\usepgfplotslibrary{statistics}
\usepackage{placeins}
\usepackage[american]{circuitikz}

\ctikzset{
	logic ports=ieee,
	logic ports/scale=0.3,
}

\tikzset{
    port/.style={
        circle,
        fill,
        inner sep=1.5pt,
        minimum size=0pt
    }
}

\usepackage{subfig}
\usepackage{pgfplots}

\usepackage{nameref}
\usepackage[obeyspaces]{url}
\usepackage[hidelinks, unicode]{hyperref} 

\usepackage[style=numeric,minbibnames=1,maxbibnames=2]{biblatex}
\usepackage[autostyle=true]{csquotes}
\usepackage{listingsutf8}
\usepackage{orcidlink}
\usepackage{siunitx}
\usepackage{xcolor}

\definecolor{sign}{HTML}{b02a2d}
\definecolor{direction}{HTML}{007900}
\definecolor{regime}{HTML}{8c399e}
\definecolor{characteristic}{HTML}{1f5dc2}
\definecolor{mantissa}{HTML}{636363}
\definecolor{error}{HTML}{BD002A}
\definecolor{cellbg}{HTML}{EDEDED}

\makeatletter
\def\lst@makecaption{%
  \def\@captype{table}%
  \@makecaption
}
\makeatother

\begin{document}

\title{Design and Implementation of a Takum Arithmetic Hardware Codec}
\titlerunning{Design and Implementation of a Takum Arithmetic Hardware Codec}
\author{Laslo Hunhold\,\orcidlink{0000-0001-8059-0298}}
\authorrunning{L. Hunhold}
\institute{%
	Parallel and Distributed Systems Group\\
	University of Cologne, Cologne, Germany\\
	\email{hunhold@uni-koeln.de}
}
\maketitle

\begin{abstract}
The takum machine number format has been recently proposed as an enhancement over the posit number format, which is considered a promising alternative to the IEEE 754 floating-point standard. Takums retain the useful posit properties, but feature a novel
exponent coding scheme that yields more precision for small and large magnitude numbers and a much higher and bounded dynamic range.
\par
This paper presents the design and implementation of a hardware codec for both takums (logarithmic number system, LNS) and linear takums (floating-point format). The codec design is emphasised, as it constitutes the primary distinguishing feature compared to logarithmic posits (LNS) and posits (floating-point format), which otherwise share similar internal representations. Furthermore, a novel internal representation for LNS is proposed. The presented takum codec, implemented in VHDL, demonstrates near-optimal scalability and performance on an FPGA. It achieves latency reductions of up to \SI{38}{\percent} and reduces LUT utilisation up to \SI{50}{\percent} compared to the best state-of-the-art posit codecs.
\end{abstract}


\section{Introduction}
For decades, IEEE 754 floating-point numbers have been the dominant format for representing real numbers in computing. However, recent demands for low- and mixed-precision arithmetic in fields such as machine learning have highlighted significant shortcomings of the IEEE 754 standard, particularly in terms of reproducibility, precision and dynamic range. While modifications within the IEEE 754 framework, such as Google's \texttt{bfloat16} \cite{bfloat16}, have been proposed to address these issues, more fundamental changes to computer arithmetic are also being explored.
One such innovation is posit arithmetic \cite{posits-beating_floating-point-2017, posits-standard-2022}, which has garnered significant attention in recent years. Posits have \textit{tapered precision} by encoding the exponent with variable length, allowing for more fraction bits around numbers near 1, at the expense of fewer fraction bits for numbers of very small or large magnitude. This encoding also allows treating posits like two's complement integers in terms of sign, ordering and negation.
\par
The posit format is generally considered simpler to implement in hardware 
compared to IEEE 754 floating-point numbers, largely due to the reduction of 
special cases. Over the past few years, several optimised hardware 
implementations have been developed. Initial work is documented in 
\cite{posit-hardware-2018, posit-hardware_cost-2019}, with further 
optimisations leading to the creation of the FloPoCo posit core generator 
\cite{flopoco-1-2020, flopoco-2-2021, flopoco-3-2021, flopoco-4-2022}. This 
culminated in the development of a complete RISC-V posit core supporting up to 
64-bit precision \cite{percival-2022, percival-64_bit-2024}.
Despite the potential of posits in various low- and mixed-precision applications \cite{posit-dnn-2019, 2024-log-posit}, several challenges have hindered their adoption. Chief among these is the sharp decline in precision for small and large magnitude numbers and the lack of dynamic range, which poses significant problems for general-purpose arithmetic applications
and numerical methods \cite{posits-good-bad-ugly-2019, 2024-takum}.
\par
In response to these challenges, the takum format was recently proposed \cite{2024-takum}. Takum introduces a new encoding scheme that guarantees a minimum of $n-12$ fraction bits and achieves a significantly increased dynamic range that is fully realised at 12-bit precision, remaining constant thereafter. This design not only makes takums suitable for applying well-established numerical analysis techniques, but it also decouples the choice of precision from dynamic range concerns in mixed-precision applications. This property is reflected in its name, derived from the
Icelandic \enquote{\emph{tak}markað \emph{um}fang}, meaning \enquote{limited range}.
\par
While the takum format has been formally verified \cite{2024-takum}, its hardware implementation has only been briefly touched upon in \cite[Section~5.2]{2024-takum}, with claims that it should be simpler than posit implementations. The primary reason is that the takum exponent is at most 11 bits long, as opposed to potentially occupying the entire width, as is the case with posits. Consequently, the bit string length in takum should have minimal impact on the complexity of encoding and decoding processes, unlike in posit arithmetic. However, the upfront computational costs of takum encoding and decoding remain an open question.
\par
This paper makes three primary contributions: (1) it derives efficient algorithms for encoding and decoding takum numbers, proposing a new logarithmic number system (LNS) internal representation; (2) it provides a detailed open-source VHDL implementation optimised for but not limited to FPGA hardware; and (3) it compares the performance of the takum codec with state-of-the-art posit codecs.
\par
The remainder of this paper is organised as follows: Section~\ref{sec:takum_definition} defines the takum format in both its logarithmic and linear forms. Section~\ref{sec:internal_representations} introduces a novel internal representation for logarithmic number systems, inspired by recent advances in posit codecs. Sections~\ref{sec:decoder} and \ref{sec:encoder} then detail the design and implementation of a takum decoder and encoder, respectively. Section~\ref{sec:evaluation} presents an evaluation of the VHDL implementation on FPGA, followed by a conclusion and outlook in Section~\ref{sec:conclusion_and_outlook}.
\section{Takum Encoding Scheme}\label{sec:takum_definition}
We begin by defining the (logarithmic) takum encoding scheme, 
originally introduced in \cite{2024-takum}, as its binary 
format and notation will be extensively referenced throughout 
this work. In terms of notation, for a bit string $B \in 
\{0,1\}$, $\overline{B}$ denotes its complement.
\begin{definition}[takum encoding {\cite[Definition~2]{2024-takum}}]\label{def:takum}
	Let $n \in \mathbb{N}$ with $n \ge 12$. Any $n$-bit MSB$\rightarrow$LSB 
	string 
	$T := (\textcolor{sign}{S},\textcolor{direction}{D},\textcolor{regime}{R},
	\textcolor{characteristic}{C},\textcolor{mantissa}{M}) \in {\{0,1\}}^n$ of 
	the form
	
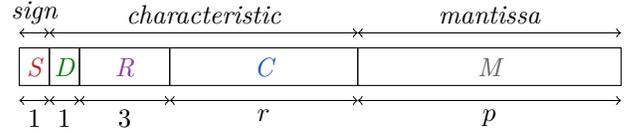
\begin{center}
		\begin{tikzpicture}
			\draw[<->] (0.0, 0.7) -- (0.4, 0.7) node[above,pos=.5] {sign};
			\draw[<->] (0.4, 0.7) -- (4.5, 0.7) node[above,pos=.5] 
			{characteristic};
			\draw[<->] (4.5, 0.7) -- (8.0, 0.7) node[above,pos=.5] {mantissa};
			
			\draw (0,  0  ) rectangle (0.4,0.5) node[pos=.5] 
			{\textcolor{sign}{S}};
			\draw (0.4,0  ) rectangle (0.8,0.5) node[pos=.5] 
			{\textcolor{direction}{D}};
			\draw (0.8,0  ) rectangle (2.0,0.5) node[pos=.5] 
			{\textcolor{regime}{R}};
			\draw (2.0,0  ) rectangle (4.5,0.5) node[pos=.5] 
			{\textcolor{characteristic}{C}};
			\draw (4.5,0  ) rectangle (8.0,0.5) node[pos=.5] 
			{\textcolor{mantissa}{M}};
			
			\draw[<->] (0.0, -0.2) -- (0.4, -0.2) node[below,pos=.5] {$1$};
			\draw[<->] (0.4, -0.2) -- (0.8, -0.2) node[below,pos=.5] {$1$};
			\draw[<->] (0.8, -0.2) -- (2.0, -0.2) node[below,pos=.5] {$3$};
			\draw[<->] (2.0, -0.2) -- (4.5, -0.2) node[below,pos=.5] {$r$};
			\draw[<->] (4.5, -0.2) -- (8.0, -0.2) node[below,pos=.5] {$p$};
		\end{tikzpicture}
	\end{center}
	with {sign bit} $\textcolor{sign}{S}$, {direction bit}
	$\textcolor{direction}{D}$, {regime bits}
	$\textcolor{regime}{R} := (\textcolor{regime}{R}_2,
	\textcolor{regime}{R}_1,\textcolor{regime}{R}_0)$, characteristic 
	bits
	$\textcolor{characteristic}{C} :=(\textcolor{characteristic}{C}_{r-1},\dots,
	\textcolor{characteristic}{C}_0)$, {mantissa bits}
	$\textcolor{mantissa}{M} := (\textcolor{mantissa}{M}_{p-1},\dots,
	\textcolor{mantissa}{M}_0)$, {regime}
	\begin{equation}
		r := \begin{cases}
			\integer(\overline{\textcolor{regime}{R}}) & 
			\textcolor{direction}{D} = 0\\
			\integer(\textcolor{regime}{R})	&
			\textcolor{direction}{D} = 1
		\end{cases}
		\in \{0,\dots,7\},
	\end{equation}
	{characteristic}
	\begin{equation}\label{eq:def-characteristic}
		c :=
		\begin{cases}
			-2^{r+1} + 1 + \integer(\textcolor{characteristic}{C}) & 
			\textcolor{direction}{D} = 
			0\\
			2^r - 1 + \integer(\textcolor{characteristic}{C})
			& \textcolor{direction}{D} = 1
		\end{cases}
		\in \{ -255,\dots,254 \},
	\end{equation}
	{mantissa bit count} $p := n - r -5 \in \{n-12,\dots,n-5\}$,
	{mantissa} $m := 2^{-p} \integer(\textcolor{mantissa}{M})
	\in [0,1)$ and {logarithmic value}
	\begin{equation}
		\ell := {(-1)}^{\textcolor{sign}{S}}(c + m) \in (-255,255)
	\end{equation}
	encodes the takum value
	\begin{equation}\label{eq:takum-value}
		\takum(T)
		:= \begin{cases}
			\begin{cases}
				0 & \textcolor{sign}{S} = 0\\
				\mathrm{NaR} & \textcolor{sign}{S} = 1
			\end{cases}
			& \textcolor{direction}{D} = \textcolor{regime}{R} = 
			\textcolor{characteristic}{C} = \textcolor{mantissa}{M} 
			= \bm{0} \\
			{(-1)}^{\textcolor{sign}{S}}
			\euler^{\ell} & \text{otherwise}
		\end{cases}
	\end{equation}
	with $\takum \colon {\{0,1\}}^n \mapsto \{ 0,\mathrm{NaR} \} \cup
	\pm\left(\euler^{-255},\euler^{255}\right)$ and \textsc{Euler}'s number
	$\mathrm{e} \approx 2.718$ (and thus
	$\euler \approx 1.649$).
	Any bit string shorter than 12 bits is also
	considered in the definition by assuming the missing bits to be
	zero bits (\enquote{ghost bits}).
\end{definition}
Takums are uniformly defined for arbitrary bit-string lengths 
$n$, unlike IEEE 754 floating-point numbers, which are only non-uniformly defined for
specific values of $n$. Although takums are primarily defined as an LNS, linear takums, their floating-point variant, retain a structure that closely mirrors the original LNS format.
\begin{definition}[linear takum encoding {\cite[Definition~20]{2024-takum}}]
	\label{def:linear_takum}
	Take Definition~\ref{def:takum} and rename
	the mantissa bit count $p$, mantissa bits
	$\textcolor{mantissa}{M}$ and mantissa $m$ to
	fraction bit count $p$,
	fraction bits $\textcolor{mantissa}{F}$ and
	fraction $f$. Define the exponent
	\begin{equation}
		e := {(-1)}^{\textcolor{sign}{S}} (c + \textcolor{sign}{S})
		\in \{ -255,\dots,254 \}
	\end{equation}
	and the linear takum value encoding as
	\begin{equation}
		\overline{\takum}(T)
		:= \begin{cases}
			\begin{cases}
				0 & \textcolor{sign}{S} = 0\\
				\mathrm{NaR} & \textcolor{sign}{S} = 1
			\end{cases}
			& \textcolor{direction}{D} = \textcolor{regime}{R} = 
			\textcolor{characteristic}{C} = \textcolor{mantissa}{F} 
			= \bm{0} \\
			[(1 - 3 \textcolor{sign}{S}) + f] \cdot 2^{e} & \text{otherwise.}
		\end{cases}
	\end{equation}
\end{definition}
As we can see, both the logarithmic and linear definitions of takum are largely similar. In this paper, we will adhere closely to both definitions for the development of the codec. Please refer to \cite{2024-takum} for an overview of the numerical properties.
\section{Internal Representations}\label{sec:internal_representations}
While it is clear that the decoder's input, as well as the 
encoder's output, is a binary representation of a takum, 
defining the decoder's output and encoder's input in general 
cases is more complex. This internal representation must not 
only uniquely identify a takum but also facilitate efficient 
arithmetic operations. Typically, decoders and encoders are 
used as input and output stages within an arithmetic unit, so 
the internal representation should be optimised for these 
contexts. In general, special cases (zero and $\mathrm{NaR}$) 
are represented separately; thus, the following discussion of 
internal representations pertains only to non-zero real numbers.
\par
In the linear case, the most straightforward approach would be to use the standard floating-point internal representation of the form
\begin{equation}\label{eq:representation-floating_point}
	\left(\textcolor{sign}{S},\hat{e},\hat{f}\right) 
	\mapsto (-1)^{\textcolor{sign}{S}} \left(1+\hat{f}\right) \cdot 
	2^{\hat{e}}
\end{equation}
with $\hat{f} \in [0,1)$ and $\hat{e} \in \{-254,\dots,254\}$.
However, as demonstrated by Yonemoto \cite{posits-beating_floating-point-2017} for the default internal representation of posits and further generalised by \cite{flopoco-4-2022}, it is more effective to adopt the internal representation
\begin{equation}\label{eq:representation-two_s_complement}
	(\textcolor{sign}{S},e,f) \mapsto [(1 - 3 \textcolor{sign}{S}) + f] \cdot 2^e
\end{equation}
with $f \in [0,1)$ and $e \in \{-255,\dots,254\}$.
This choice is advantageous because the internal representation is monotonic in $f$, similar to how the fraction bits
$\textcolor{mantissa}{F}$ are monotonic in takum encoding. This 
monotonicity avoids the need for a full two's complement 
negation of the fraction during decoding and encoding when 
$\textcolor{sign}{S}$ is 1, as would be necessary with the 
representation in (\ref{eq:representation-floating_point}). The 
trade-off between (\ref{eq:representation-floating_point}) and 
(\ref{eq:representation-two_s_complement}) involves a slightly 
more complex arithmetic for the latter and a departure from the 
extensive body of work associated with the former, as discussed 
in \cite{flopoco-3-2021, flopoco-4-2022}. Moreover, examining 
the exponent definition, $e = (-1)^{\textcolor{sign}{S}} (c + 
\textcolor{sign}{S})$, reveals that it simplifies to $c$ when 
$\textcolor{sign}{S} = 0$, and to the bitwise complement of 
$c$, assuming $c$ is represented as a two's complement integer, 
when $\textcolor{sign}{S} = 1$.
\par
Turning to the logarithmic case, the canonical internal representation is given by
\begin{equation}\label{eq:representation-logarithmic}
	(\textcolor{sign}{S},\ell) \mapsto {(-1)}^{\textcolor{sign}{S}} \euler^{\ell}
\end{equation}
where $\ell \in (-255, 255)$. This representation is also utilised in the most recent work on logarithmic posits \cite{2024-log-posit}. However, this approach has a drawback similar to the linear case, as it necessitates a two's complement negation both after decoding to and before encoding from the internal representation.
By definition, $\ell = (-1)^{\textcolor{sign}{S}} (c + m)$. An alternative we propose here is to define a \enquote{barred logarithmic value} as $\overline{\ell} := c + m = (-1)^{\textcolor{sign}{S}} \ell \in (-255, 255)$ and use the internal representation
\begin{equation}\label{eq:representation-barred_logarithmic}
	\left(\textcolor{sign}{S},\overline{\ell}\right) \mapsto {(-1)}^{\textcolor{sign}{S}} \euler^{{(-1)}^{\textcolor{sign}{S}} \overline{\ell}}.
\end{equation}
The advantage of this novel representation is that $\overline{\ell}$ is monotonic in $m$, just as the mantissa bits $\textcolor{mantissa}{M}$ are monotonic in takum encoding. This monotonicity eliminates the need for two's complement negations after decoding and before encoding. The impact on arithmetic complexity is minimal, as all sign cases of $\ell$ must be handled regardless. For instance, when computing the square root, which involves a right shift of $\ell$, the procedure remains unchanged. Similarly, multiplication and division require handling all sign combinations of $\ell$, as does addition and subtraction.
\par
The barred logarithmic value $\overline{\ell}$ can be easily derived from $c$ and $m$ by concatenating the 9-bit signed integer $c$ with the $(n-5)$-bit unsigned integer $2^{n-5} \cdot m$ (given $p + r = n - 5$), which is then interpreted as an $(n+4)$-bit fixed-point number.
In general, whether in the linear or logarithmic case, the internal representations $(\textcolor{sign}{S}, e, f)$ and $(\textcolor{sign}{S}, \overline{\ell})$ are straightforward to determine from the characteristic $c$ and the equivalent
fraction $f$ and mantissa $m$.
\section{Decoder}\label{sec:decoder}
The purpose of the decoder is to transform a given $n$-bit 
takum binary representation, where $n \in \mathbb{N}_2$
(the set of natural numbers greater than or equal to 2), into 
its corresponding internal representation. Specifically, this 
internal representation is expressed as $(\textcolor{sign}{S}, 
e, f)$ in the linear case and as $(\textcolor{sign}{S}, 
\overline{\ell})$ in the logarithmic case (see 
(\ref{eq:representation-two_s_complement}) and 
(\ref{eq:representation-barred_logarithmic})). As demonstrated 
in Section~\ref{sec:internal_representations}, both forms of 
internal representation can be directly derived from the 
characteristic $c$ and the mantissa $m$ (where the fraction $f$ 
corresponds to $m$). Consequently, our primary objective is to 
determine $c$ and $m$, as they serve as the common foundation 
for both cases to derive the respective internal 
representation. Additionally, if the input corresponds to $0$ 
or $\mathrm{NaR}$, this should also be appropriately flagged.
\par
Given that takums are a tapered precision machine number format, it is also essential to determine the number of mantissa bits (precision) for a given input. The direction bit $\textcolor{direction}{D}$ ($\mathit{direction\_bit}$) and the sign bit $\textcolor{sign}{S}$ ($\mathit{sign}$) can be straightforwardly extracted from the two most significant bits (MSBs) of the input.
\subsection{Characteristic/Exponent Determinator}
To determine the characteristic (and later, the exponent), we focus on the following problem reduction:
Referring to the takum bit pattern defined in Definition~\ref{def:takum} and considering that both the characteristic length and the bit pattern are variable, it is most effective to expand the bit pattern to 12 bits and then extract the 7 bits that follow the regime. This extraction will be carried out later; our current focus is on the processing of these bits.
These 7 bits, which we refer to as the raw characteristic 
($\mathit{characteristic\_raw\_bits}$), are guaranteed to include all 
$\textit{regime}$ characteristic bits, followed by $7 - \mathit{regime} =: 
\textit{antiregime}$ waste bits.
\par
If we examine the definition of the characteristic provided in (\ref{eq:def-characteristic}), we observe that it involves introducing a bias to the characteristic bits. This process entails a bitwise operation, followed by either an increment or a decrement of the bits. To simplify the procedure and avoid the need for decrements, we propose a normalisation approach that exclusively relies on increments. To justify this approach, we first present the following
\begin{proposition}[characteristic 
complement]\label{prop:characteristic_complement}
	Let $n \in \mathbb{N}_1$ and bit string
	$T := (\textcolor{sign}{S},\textcolor{direction}{D},\textcolor{regime}{R},
	\textcolor{characteristic}{C},\textcolor{mantissa}{M})
	\in {\{0,1\}}^n$ as in Definition~\ref{def:takum} with
	$\takum((\textcolor{sign}{S},\textcolor{direction}{D}, \textcolor{regime}{R},
	\textcolor{characteristic}{C},\textcolor{mantissa}{M})) \notin \{0,\mathrm{NaR} \}$,
	characteristic $c$ and regime $r$. Let
	$\tilde{T} := (\textcolor{sign}{\tilde{S}},\overline{\textcolor{direction}{D}},\overline{\textcolor{regime}{R}},
	\overline{\textcolor{characteristic}{C}},\textcolor{mantissa}{\tilde{M}}) \in {\{0,1\}}^n$
	with characteristic $\tilde{c}$, regime $\tilde{r}$
	and arbitrary $\textcolor{sign}{\tilde{S}}$ and $\textcolor{mantissa}{\tilde{M}}$.
	It holds $\tilde{r} = r$ and $\tilde{c} = -c - 1$.
	This means that, assuming $c$ is represented as a two's 
	complement integer, $\tilde{c}$ is obtained directly as 
	the bitwise complement of $c$.
\end{proposition}
\begin{proof}
	It follows by definition that $\tilde{r} = r$ given the 
	complementation of the direction bit cancels out that 
	of the regime bits, yielding the same regime value. It 
	holds for the characteristic value
	\begin{align}
		\tilde{c} &= \sum_{i=0}^{r-1} \overline{\textcolor{characteristic}{C}}_i 2^i +
			\begin{cases}
				-2^{r+1} + 1 & \overline{\textcolor{direction}{D}} = 0\\
				2^{r} - 1 & \overline{\textcolor{direction}{D}} = 1\\
			\end{cases}\\
		&= (2^r - 1) - \sum_{i=0}^{r-1} \textcolor{characteristic}{C}_i 2^i +
			\begin{cases}
				-2^{r+1} + 1 & \overline{\textcolor{direction}{D}} = 0\\
				2^{r} - 1 & \overline{\textcolor{direction}{D}} = 1\\
			\end{cases}\\
		&= -\left(
				 \sum_{i=0}^{r-1} \textcolor{characteristic}{C}_i 2^i +
				\begin{cases}
					2^{r} - 1 & \overline{\textcolor{direction}{D}} = 0\\
					-2^{r+1} + 1 & \overline{\textcolor{direction}{D}} = 1\\
				\end{cases}
			\right) - 1\\
		&= -c - 1.
	\end{align}
	Thus we have proven what was to be shown.
	\qed
\end{proof}
As observed, negating the direction bit, the regime bits, and the characteristic bits corresponds to negating the characteristic value, which is represented in two's complement. This property can be leveraged in the following manner:
\begin{corollary}[conditional characteristic complement]
	\label{cor:conditional_characteristic_complement}
	Let $n \in \mathbb{N}_1$ and
	$T := (\textcolor{sign}{S},\textcolor{direction}{D},\textcolor{regime}{R},
	\textcolor{characteristic}{C},\textcolor{mantissa}{M})
	\in {\{0,1\}}^n$ as in Definition~\ref{def:takum} with
	$\takum((\textcolor{sign}{S},\textcolor{direction}{D}, \textcolor{regime}{R},
	\textcolor{characteristic}{C},\textcolor{mantissa}{M})) \notin \{0,\mathrm{NaR} \}$,
	characteristic $c$ and regime $r$. Let
	\begin{equation}
		\tilde{T} := \begin{cases}
			(\textcolor{sign}{S},\textcolor{direction}{D},\textcolor{regime}{R},
				\textcolor{characteristic}{C},\textcolor{mantissa}{M}) & \textcolor{direction}{D} = 0\\
			(\textcolor{sign}{S},\overline{\textcolor{direction}{D}},\overline{\textcolor{regime}{R}},
								\overline{\textcolor{characteristic}{C}},\textcolor{mantissa}{M}) & \textcolor{direction}{D} = 1
		\end{cases}
	\end{equation}
	with characteristic $\tilde{c}$ and regime $\tilde{r}$. It holds $\tilde{r} = r$ and
	\begin{equation}
		\tilde{c} =
		\begin{cases}
			c \!& \textcolor{direction}{D}\!=\!0\\
			-c-1 \!& \textcolor{direction}{D}\!=\!1
		\end{cases} =
		\begin{cases}
			-2^{r+1} \!+ \!1 \!+ \!\sum_{i=0}^{r-1} \textcolor{characteristic}{C}_i 2^i \!& \textcolor{direction}{D}\!=\!0\\
			-2^{r+1} \!+ \!1 \!+ \!\sum_{i=0}^{r-1} \overline{\textcolor{characteristic}{C}}_i 2^i \!& \textcolor{direction}{D}\!=\!1.
		\end{cases}
	\end{equation}
	It should be noted that $c$ can be obtained from $-c-1$ 
	by a simple bitwise complement, assuming $c$ is 
	represented as a two's complement integer.
\end{corollary}
\begin{proof}
	This follows directly from 
	Proposition~\ref{prop:characteristic_complement} and
	Definition~\ref{def:takum}.
	\qed
\end{proof}
This approach suggests a strategy for determining the characteristic $c$: First, conditionally negate the characteristic bits when the direction bit is 1. Next, apply the bias $-2^{r+1}$ unconditionally, then increment the result. Finally, conditionally negate the intermediate result $\tilde{c}$ when the direction bit is 1 to obtain the final characteristic $c$. This ensures that the process involves only incrementing in all cases (without any decrement), allowing us to focus solely on the single application of the bias $-2^{r+1}$.
\par
Fortunately, the application of the bias is straightforward: As 
illustrated in Table~\ref{tab:encoder-biases}, each of the 8 
possible biases can be added to the $r$ characteristic bits, 
assuming they are represented as a two's complement integer, 
using a simple bitwise OR operation. Additionally, since the 
$(r+1)$th bit is always zero in the biases, the bias can also 
be applied using a bitwise OR operation on the incremented 
characteristic bits.
\begin{table}[tbp]
	\caption{
		All possible biases $-2^{r+1}$ given as 9-bit two's complement integers 
		for
		all regimes $r$ ranging from $0$ to $7$. The $r$ least significant bits
		are underlined.
	}
	\label{tab:encoder-biases}
	\centering
	\bgroup
	\def\arraystretch{1.2}
	\setlength{\tabcolsep}{0.5em}
	\begin{tabular}{| c | c |}
		\hline
		$r$ & $-2^{r+1}$\\\hline\hline
		0 & \texttt{111111110} \\\hline
		1 & \texttt{11111110\underline{0}} \\\hline
		2 & \texttt{1111110\underline{00}} \\\hline
		3 & \texttt{111110\underline{000}} \\\hline	
	\end{tabular}
	\hspace{0.4cm}
	\begin{tabular}{| c | c |}
		\hline
		$r$ & $-2^{r+1}$\\\hline\hline
		4 & \texttt{11110\underline{0000}} \\\hline
		5 & \texttt{1110\underline{00000}} \\\hline
		6 & \texttt{110\underline{000000}} \\\hline
		7 & \texttt{10\underline{0000000}} \\\hline		
	\end{tabular}
	\egroup
\end{table}
\par
However, we adopt a different approach: Given that our raw characteristic contains the left-aligned characteristic bits, we first invert these bits if the direction bit is $1$. Next, we prepend the bits \texttt{10} to the left and perform an \emph{arithmetic} right shift by the \textit{antiregime} bits. This process yields the desired bitwise OR of the bias with the characteristic bits.
Subsequently, we increment the first 8 bits of the resulting 
value, prepend \texttt{1} to the left, and then conditionally 
negate the outcome based on the direction bit, as illustrated 
in Corollary~\ref{cor:conditional_characteristic_complement}.
\par
Referring back to the previously discussed internal representations in 
Section~\ref{sec:internal_representations}, it was noted that in the case of 
linear takums, the exponent is obtained by negating the characteristic. Rather 
than performing this negation separately, which would introduce additional 
overhead, we introduce an additional input bit, $\mathit{output\_exponent}$. 
This parameter, which is set during synthesis and not a true input bit of the 
entity, inverts the conditional negation we already perform, thus providing the 
exponent negation at no extra cost.
\subsection{Pre-, Logarithmic and Linear Decoder}
\begin{figure}[tbp]
	\begin{center}
		\begin{circuitikz}[scale=0.8, transform shape]
			\node
			    [draw, rectangle, densely dotted, thick, minimum width=2cm, minimum height=3cm,
			    label={[anchor=south east]south east:{E1}}]
			    (E1) at (-1.1,0.5) {};

			\node
			    [draw, rectangle, densely dotted, thick, minimum width=3.9cm, minimum height=4.6cm,
			    label={[anchor=south east]south east:{E2}}]
			    (E2) at (3.65,0.5) {};

			\node
			    [draw, rectangle, densely dotted, thick, minimum width=4cm, minimum height=2.3cm,
			    label={[anchor=south east]south east:{E3}}]
			    (E3) at (-3.0,-5.25) {};

			\node[port, label={[anchor=south]above:{$\mathit{takum}$}}] 
			    (in1) at (-6,-1) {};
			\node(helper1) at ([xshift=1cm,yshift=-1.6cm]in1) {};
			\node(helper2) at ([xshift=9cm]E1.center) {};

			\node[port] 
			    (E1in1) at ([yshift=-0.8cm]E1.north west) {};
			\node[port] 
			    (E1in2) at ([yshift=-1.6cm]E1.north west) {};
			\node[port, label={[anchor=south west]above right:{$\mathit{antiregime}$}}] 
			    (E1out1) at ([yshift=-0.8cm]E1.north east) {};
			\node[port, label={[anchor=south west]above right:{$\mathit{regime}$}}] 
			    (E1out2) at ([yshift=-1.6cm]E1.north east) {};

			\node[port] 
				(E2in0) at ([yshift=-0.8cm]E2.north west) {};
			\node[port] 
				(E2in1) at ([yshift=-1.6cm]E2.north west) {};
			\node[port] 
				(E2in2) at ([yshift=-2.4cm]E2.north west) {};
			\node[port, label={[anchor=south west,align=left]above right:{$\mathit{characteristic\_}$\\$\mathit{or\_exponent}$}}] 
			    (E2out1) at ([yshift=-0.8cm]E2.north east) {};

			\node[port] 
				(E3in1) at ([yshift=-0.8cm]E3.north west) {};
			\node[port, label={[anchor=south west]above right:{$\mathit{is\_zero}$}}] 
			    (E3out1) at ([yshift=-0.8cm]E3.north east) {};
			\node[port, label={[anchor=south west]above right:{$\mathit{is\_nar}$}}] 
			    (E3out2) at ([yshift=-1.6cm]E3.north east) {};

			\draw ([xshift=-0.5cm]E1out1) node[scale=0.5,iampshape,fill=white,t={\ \,MUX}] (E1mux1) {};
			\draw ([xshift=-0.5cm]E1out2) node[scale=0.5,iampshape,fill=white,t={\ \,MUX}] (E1mux2) {};
			\draw node[ieeestd not port, no input leads, no output leads] at ([xshift=0.6cm,yshift=-0.8cm]E1in2.center) (E1not) {};
			\draw[very thick] ([xshift=-0.05cm]E1not.east) -- ++(0.2,0) |- ([yshift=-0.15cm]E1mux1.west);
			\draw[very thick] ([yshift=0.15cm]E1mux2.west) --
				++(-0.25,0);
			\draw[very thick] (E1in2) -- ([xshift=0.7cm]E1in2.center) |-
				([yshift=0.15cm]E1mux1.west);
			\draw[very thick] ([yshift=-0.15cm]E1mux2.west) -|
				([xshift=0.7cm]E1in2.center |- E1in2);
			\draw[very thick] ([xshift=0.3cm]E1in2.center) |- (E1not.west);

			\draw (E1in1) -| ++(0.3,-0.4) -| ([yshift=-0.05cm]E1mux2.north);
			\draw ([xshift=0.3cm]E1in1.center) |-
				++(0.5,0.5) -|
				([yshift=-0.05cm]E1mux1.north);

			\draw[very thick] (E1mux1.east) -- (E1out1.center);
			\draw[very thick] (E1mux2.east) -- (E1out2.center);
			
			\draw[very thick] (E2in2.center) -|
				node[midway,fill=white,yshift=-0.6cm,inner sep=2pt]{$7$}
				++(0.3cm,-1.8cm) --
				++(0.3,0)
				node (E2notansatz) {};
			\draw node[ieeestd not port, no input leads, no output leads] at ([xshift=0.15cm]E2notansatz) (E2not) {};
			\draw[very thick] ([xshift=-0.05cm]E2not.east) -- ++(0.3,0)
			node (E2muxansatz) {};
			\draw ([xshift=0.25cm,yshift=0.15cm]E2muxansatz) node[scale=0.5,iampshape,fill=white,t={\ \,MUX}] (E2mux1) {};
			\draw[very thick] ([yshift=0.15cm]E2mux1.west) --
				++(-0.95,0);
			\draw (E2in0.center) -| ++(0.6,0) |-
				++(0,-2.8) -|
				([yshift=-0.05cm]E2mux1.north);
			\draw[very thick] (E2mux1.east) -- ++(0.2,0) node (E2fusion) {};
			\draw[densely dashed] (E2fusion.center) -- ([yshift=0.8cm]E2fusion.center)
				node (E2fusiontop) {};
			\draw[very thick] (E2fusiontop.center) --
				node[midway,fill=white,inner sep=1pt]{$2$}
				++(-0.6,0)
				node[port,label={[anchor=east]left:{\texttt{10}}}] {};
			\draw[very thick] ([yshift=0.4cm]E2fusion.center) --
				node[midway,fill=white,inner sep=1pt]{$9$}
				++(0.5,0) node (E2shiftr1ansatz) {};
			\node[draw, scale=0.6, fill=white, minimum width=1.5cm, minimum height=1cm,
			      rounded corners=3pt, thick] (E2shiftr1) at ([xshift=0.45cm]E2shiftr1ansatz.center) {ASHIFTR};
			\draw[very thick] (E2in1.center) -| ++(1.0,-1.3) -| (E2shiftr1.north);
			\draw[very thick] (E2shiftr1.east) --
				++(0.3,0) |-
				node[midway,fill=white,xshift=-1cm,inner sep=1pt]{$9$}
				++(-2.3,1.1) |- ++(-0.6,1.3) |- ++(0.3,0.3) node (E2fusion2) {};
			\draw[densely dashed] (E2fusion2.center) -- ([yshift=0.7cm]E2fusion2.center)
				node (E2fusion2top) {};
			\draw (E2fusion2top.center) --
				++(0.4,0)
				node[port,label={[anchor=west]right:{\texttt{1}}}] {};
			\draw[very thick] ([yshift=0.2cm]E2fusion2.center) --
				node[midway,fill=white,inner sep=1pt]{$8$}
				++(0.5,0) node (E2incansatz) {};
			\node[draw, scale=0.6, fill=white, minimum width=1.5cm, minimum height=1cm,
			      rounded corners=3pt, thick] (E2inc) at ([xshift=0.45cm]E2incansatz.center) {INC};
			\draw[very thick] (E2inc.east) -|
				node[midway,xshift=-0.26cm,fill=white,inner sep=1pt]{$8$}
				++(0.5,0) node (E2fusion3) {};
			\draw[densely dashed] ([yshift=-0.2cm]E2fusion3.center) --
				([yshift=0.5cm]E2fusion3.center)
				node (E2fusion3top) {};
			\draw (E2fusion3top.center) --
				++(-0.4,0)
				node[port,label={[anchor=east]left:{\texttt{1}}}] {};
			\draw[very thick] ([yshift=-0.2cm]E2fusion3.center) -|
				++(0.3,-0.3) -|
				node[midway,xshift=0.75cm,fill=white,inner sep=1pt]{$9$}
				++(-1.5,-0.9) -- ++(0.35,0) node (E2not2ansatz) {};
			\draw node[ieeestd not port, no input leads, no output leads]
				at ([xshift=0.15cm]E2not2ansatz) (E2not2) {};
			\draw[very thick] ([xshift=-0.05cm]E2not2.east) -- ++(0.3,0)
				node (E2mux2ansatz) {};
			\draw ([xshift=0.25cm,yshift=0.15cm]E2mux2ansatz)
				node[scale=0.5,iampshape,fill=white,t={\ \,MUX}] (E2mux2) {};
			\draw[very thick] ([yshift=0.15cm]E2mux2.west) -- ++(-1,0);
			\draw ([yshift=-0.05cm]E2mux2.north) |- ++(-2.5,0.37);
			\draw[very thick] (E2mux2.east) -- ++(0.25,0) |- (E2out1.center);
			
			\node[and port] (E3and1) at ([xshift=-0.6cm]E3out1) {};
			\node[and port] (E3and2) at ([xshift=-0.6cm]E3out2) {};
			\draw (E3and1.out) -- ++(0.2cm,0);
			\draw (E3and2.out) -- ++(0.2cm,0);

			\draw[very thick] (E3in1.center) -- ++(0.5,0);
			\draw[densely dashed] ([xshift=0.5cm,yshift=0.4cm]E3in1.center) node (E3fusiontop) {} --
				(E3fusiontop.center |- E3and2.in 2) node (E3fusionbot) {};
			\draw (E3fusiontop.center) -| ++(1.0,-0.7) node (E3dropbridgebot) {};

			\node (E3dropbridgemid) at (E3dropbridgebot |- E3and1.in 1) {};
			\node[not port, no output leads] (E3inv) at ([xshift=0.661cm]E3dropbridgemid) {};
			\draw (E3dropbridgemid.center |- E3inv.in) -- (E3inv.in);
			\draw ([xshift=-0.05cm]E3inv.out |- E3and1.in 1) -- (E3and1.in 1);
			\draw (E3dropbridgebot.center) -| ([xshift=-0.2cm]E3and2.in 1) --
				(E3and2.in 1);
			\draw[very thick]
				(E3fusionbot.center) --
				node[midway,fill=white,inner sep=1pt]{$n\!-\!1$}
				++(1.4,0) node (E3noransatz) {};
			\node[nor port,no input leads,no output leads] (E3nor) at ([xshift=0.2cm]E3noransatz) {};

			\draw ([xshift=-0.05cm]E3nor.out |- E3and2.in 2) -- (E3and2.in 2);
			\draw ([xshift=0.2cm]E3nor.out) |- (E3and1.in 2);
			
			\draw[very thick] (in1) --
				node[midway,fill=white,inner sep=1pt]{$n$}
				++(1,0);
			\draw[densely dashed] (helper1.center |- 0,3.1) --
				(helper1.center);
			\draw ([xshift=-0.01cm]helper1.center |- 0,3.1) --
				++(1.41,0) node[port,label={[anchor=south]above:{$\mathit{sign\_bit}$}}] {} --
				(helper2 |- 0,3.1) node[port]{};
			\draw (helper1.center |- E1in1) --
				++(1.4,0) node[port,label={[xshift=-0.3cm,anchor=south]above:{$\mathit{direction\_bit}$}}] {} --
				(E1in1);
			\draw[very thick] (helper1.center |- E1in2) --
				node[midway,fill=white,inner sep=1pt]{$3$}
				++(1.4,0) node[port,label={[xshift=-0.3cm,anchor=south]above:{$\mathit{regime\_bits}$}}] {} --
				(E1in2);
			\draw[very thick] ([yshift=0.8cm]helper1.center) --
				node[midway,fill=white,inner sep=1pt]{$7$}
				++(1.4,0)
				node[port,label={[xshift=1.3cm,anchor=south]above :{$\mathit{characteristic\_raw\_bits}$}}] (crb) {} -|
				([xshift=1.2cm]E1.south east) |-
				(E2in2);
			\draw[very thick] ([xshift=-0.01cm]helper1.center) --
				node[midway,fill=white,inner sep=1pt]{$n\!-\!12$}
				++(1.41,0) node[port]{} -- ++(1,0);

			\draw[densely dashed] ([xshift=1cm]crb.center) -- ++(0,-1.2) node (mantissastart) {};
			
			\draw[very thick] ([xshift=-0.01cm]mantissastart.center) --
				node[midway,xshift=-1cm,fill=white,inner sep=1pt]{$n\!-\!5$}
				++(5.12,0) node (mantissashiftstart) {};

			\node[draw, scale=0.6, fill=white, minimum width=1.5cm, minimum height=1cm,
			      rounded corners=3pt, thick] (shiftl) at (mantissashiftstart) {SHIFTL};
			\draw[very thick] ([xshift=0.6cm,yshift=-2.8cm]E1out2.center) -| (shiftl.north);
			\draw[very thick] (shiftl.east) --
				node[xshift=-0.5cm,midway,fill=white,inner sep=1pt]{$n\!-\!5$}
				(helper2 |- shiftl.east)
				node[port,label={[anchor=south east]above left:{$\mathit{mantissa\_bits}$}}] {};

			\node[draw, scale=0.6, fill=white, minimum width=1.5cm, minimum height=1cm,
			      rounded corners=3pt, thick] (sub) at ([yshift=-1cm]mantissashiftstart.center) {SUB};
			\draw[very thick] (E1out2) --
				([xshift=0.6cm]E1out2.center) |-
				node[midway,yshift=3.5cm,fill=white,inner sep=2pt]{$3$}
				([yshift=-0.2cm]sub.west);
			\draw[very thick] ([yshift=0.2cm]sub.west) --
				++(-0.5cm,0) node[port,label={[anchor=east]left:{$n\!-\!5$}},inner sep=1pt]{};
			\draw[very thick] (sub.east) -- (helper2 |- sub.east)
				node[port,label={[anchor=south east]above left:{$\mathit{precision}$}}] {};

			\draw ([xshift=-0.5cm]E1in1.center) |- ([xshift=0.9cm]E1out1 |- 0,2.5) |- (E2in0);

			\draw[very thick] (E1out1) -- node[midway,fill=white,inner sep=1pt]{$3$} (E2in1);
			
			\draw[very thick] (E2out1) --
				node[midway,fill=white,inner sep=1pt]{$9$}
				(helper2 |- E2out1) node[port] {};

			\draw[very thick] (in1) --
				node[midway,fill=white,inner sep=2pt]{$n$} (in1 |- E3in1) --
				(E3in1);
			\draw (E3out1) -- (helper2 |- E3out1) node[port]{};
			\draw (E3out2) -- (helper2 |- E3out2) node[port]{};
		\end{circuitikz}
	\end{center}
	\caption{
		The logic circuit of the predecoder, largely separated into three main entities: the regime/antiregime determinator (E1), the characteristic/exponent determinator (E2) and the special case detector (E3). We assume $n \ge 12$ (thus omitting optional zero-expansion of $\mathit{takum}$ at the beginning for
		$n < 12$) for simplicity; the implemented 
		predecoder works for any $n \ge 2$. 
		We also assume an enabled 
		$\mathit{output\_exponent}$, as disabling it 
		would only flip the top MUX in E2. Vertical 
		dashed lines indicate where the strands of a 
		multi-signal are split up or combined.
	}
	\label{fig:schematic-predecoder}
\end{figure}
We integrate the characteristic/exponent determinator into the predecoder to determine the sign, characteristic or exponent (depending on $\mathit{output\_exponent}$), and mantissa. Both parameters $n$ and $\mathit{output\_exponent}$ are already set during synthesis, thereby eliminating many conditional operations.
In total the predecoder is made up of three entities: the regime/antiregime determinator (E1), the previously described characteristic/exponent determinator (E2) and the special cases detector (E3) \cite[rtl/{\allowbreak}decoder/{\allowbreak}predecoder.vhd, lines 74--133]{code}.
A full logic circuit is outlined in Figure~\ref{fig:schematic-predecoder}.
\par
For the decoding of logarithmic takums the predecoder output is 
utilised to concatenate the characteristic and mantissa, 
forming the barred logarithmic value as defined in
\eqref{eq:representation-barred_logarithmic} 
\cite[rtl/{\allowbreak}decoder/{\allowbreak}decoder\_logarithmic.vhd]{code}.
The linear decoder, detailed in 
\cite[rtl/{\allowbreak}decoder/{\allowbreak}decoder\_linear.vhd]{code}, 
leverages the predecoder to produce the internal representation described by 
(\ref{eq:representation-two_s_complement}). In this context, the input bit 
parameter $\mathit{output\_exponent}$ is set to $1$ within the predecoder to 
directly yield the exponent instead of the characteristic.
\section{Encoder}\label{sec:encoder}
The objective of the encoder is to transform a given internal representation into the takum binary representation for a specified $n \in \mathbb{N}_2$. The internal representation is given by $(\textcolor{sign}{S}, e, f)$ in the linear case and by $(\textcolor{sign}{S}, \overline{\ell})$ in the logarithmic case, as defined in (\ref{eq:representation-two_s_complement}) and (\ref{eq:representation-barred_logarithmic}), respectively.
\par
As discussed in Section~\ref{sec:internal_representations}, both internal representations can be readily converted into the characteristic $c$ and the mantissa $m$ (where the fraction $f$ corresponds to $m$). Therefore, we can consider both representations as equivalent starting points for either case.
The direction bit $\textcolor{direction}{D}$ (denoted as $\mathit{direction\_bit}$) can be straightforwardly determined by noting that it takes the value $1$ when the characteristic satisfies $\mathit{characteristic} \ge 0$.
\subsection{Underflow/Overflow Predictor}\label{subsec:encoder-predictor}
Early in the process, it is crucial to predict whether the 
value we intend to encode might result in either an underflow 
or overflow, particularly since we are adhering to sticky 
arithmetic. In sticky arithmetic, such under- or overflows do 
not yield zero or infinity, but instead saturate at the 
smallest or largest representable number, respectively. 
Anticipating these potential outcomes 
early on enables us to significantly reduce the time on the 
critical path, as opposed to the alternative method of 
unconditionally rounding at the end and subsequently checking 
for overflow or underflow.
\par
To predict rounding, we can differentiate between two distinct 
cases that are both handled in the encoder. The first case 
occurs when $n$ lies between $2$ and $11$. In this scenario, 
the rounding boundary is positioned within the regime or 
characteristic. As detailed in Table~\ref{tab:rounding}, we can 
see that under- and overflow can be predicted with the given 
characteristic.
\begin{table}[tbp]
	\caption{
		Bit pattern and subsequent characteristic bounds where rounding down would underflow to zero or rounding up would
		overflow to $\mathrm{NaR}$ for the special cases $n \in \{ 2,\dots,11 \}$.
		The rounding boundary is indicated with a small gap.
	}
	\label{tab:rounding}
	\centering
	\bgroup
	\def\arraystretch{1.2}
	\setlength{\tabcolsep}{0.5em}
	\begin{tabular}{| c || l | c || l | c |}
		\hline
		$n$ & round-down underflows $\le$ & $c \le$ & round-up overflows $\ge$ & c $\ge$\\\hline\hline
		2 &
			\texttt{\textcolor{sign}{X}\textcolor{direction}{0}%
				\,\textcolor{regime}{111}\textcolor{characteristic}{}\textcolor{mantissa}{1\dots1}}
			 	& -1 &
				\texttt{\textcolor{sign}{X}\textcolor{direction}{1}%
				\,\textcolor{regime}{000}\textcolor{characteristic}{}\textcolor{mantissa}{0\dots0}}
				& 0 \\\hline
		3 &
			\texttt{\textcolor{sign}{X}\textcolor{direction}{0}%
				\textcolor{regime}{0}\,\textcolor{regime}{11}\textcolor{characteristic}{1111}\textcolor{mantissa}{1\dots1}}
			 	& -16 &
				\texttt{\textcolor{sign}{X}\textcolor{direction}{1}%
				\textcolor{regime}{1}\,\textcolor{regime}{00}\textcolor{characteristic}{000000}\textcolor{mantissa}{0\dots0}}
				& 15 \\\hline
		4 &
			\texttt{\textcolor{sign}{X}\textcolor{direction}{0}%
				\textcolor{regime}{00}\,\textcolor{regime}{1}\textcolor{characteristic}{111111}\textcolor{mantissa}{1\dots1}}
			 	& -64 &
				\texttt{\textcolor{sign}{X}\textcolor{direction}{1}%
				\textcolor{regime}{11}\,\textcolor{regime}{0}\textcolor{characteristic}{0000000}\textcolor{mantissa}{0\dots0}}
				& 63 \\\hline
		5 &
			\texttt{\textcolor{sign}{X}\textcolor{direction}{0}%
				\textcolor{regime}{000}\,\textcolor{characteristic}{1111111}\textcolor{mantissa}{1\dots1}}
			 	& -128 &
				\texttt{\textcolor{sign}{X}\textcolor{direction}{1}%
				\textcolor{regime}{111}\,\textcolor{characteristic}{0000000}\textcolor{mantissa}{0\dots0}}
				& 127\\\hline
		6 &
			\texttt{\textcolor{sign}{X}\textcolor{direction}{0}%
				\textcolor{regime}{000}%
				\textcolor{characteristic}{0}\,\textcolor{characteristic}{111111}\textcolor{mantissa}{1\dots1}}
			 	& -192 &
				\texttt{\textcolor{sign}{X}\textcolor{direction}{1}%
				\textcolor{regime}{111}%
				\textcolor{characteristic}{1}\,\textcolor{characteristic}{000000}\textcolor{mantissa}{0\dots0}}
				& 191 \\\hline
		7 &
			\texttt{\textcolor{sign}{X}\textcolor{direction}{0}%
				\textcolor{regime}{000}%
				\textcolor{characteristic}{00}\,\textcolor{characteristic}{11111}\textcolor{mantissa}{1\dots1}}
			 	& -224 &
				\texttt{\textcolor{sign}{X}\textcolor{direction}{1}%
				\textcolor{regime}{111}%
				\textcolor{characteristic}{11}\,\textcolor{characteristic}{00000}\textcolor{mantissa}{0\dots0}}
				& 223 \\\hline
		8 &
			\texttt{\textcolor{sign}{X}\textcolor{direction}{0}%
				\textcolor{regime}{000}%
				\textcolor{characteristic}{000}\,\textcolor{characteristic}{1111}\textcolor{mantissa}{1\dots1}}
			 	& -240 &
				\texttt{\textcolor{sign}{X}\textcolor{direction}{1}%
				\textcolor{regime}{111}%
				\textcolor{characteristic}{111}\,\textcolor{characteristic}{0000}\textcolor{mantissa}{0\dots0}}
				& 239 \\\hline
		9 &
			\texttt{\textcolor{sign}{X}\textcolor{direction}{0}%
				\textcolor{regime}{000}%
				\textcolor{characteristic}{0000}\,\textcolor{characteristic}{111}\textcolor{mantissa}{1\dots1}}
			 	& -248 &
				\texttt{\textcolor{sign}{X}\textcolor{direction}{1}%
				\textcolor{regime}{111}%
				\textcolor{characteristic}{1111}\,\textcolor{characteristic}{000}\textcolor{mantissa}{0\dots0}}
				& 247 \\\hline
		10 &
			\texttt{\textcolor{sign}{X}\textcolor{direction}{0}%
				\textcolor{regime}{000}%
				\textcolor{characteristic}{00000}\,\textcolor{characteristic}{11}\textcolor{mantissa}{1\dots1}}
			 	& -252 &
				\texttt{\textcolor{sign}{X}\textcolor{direction}{1}%
				\textcolor{regime}{111}%
				\textcolor{characteristic}{11111}\,\textcolor{characteristic}{00}\textcolor{mantissa}{0\dots0}}
				& 251 \\\hline
		11 &
			\texttt{\textcolor{sign}{X}\textcolor{direction}{0}%
				\textcolor{regime}{000}%
				\textcolor{characteristic}{000000}\,\textcolor{characteristic}{1}\textcolor{mantissa}{1\dots1}}
			 	& -254 &
				\texttt{\textcolor{sign}{X}\textcolor{direction}{1}%
				\textcolor{regime}{111}%
				\textcolor{characteristic}{111111}\,\textcolor{characteristic}{0}\textcolor{mantissa}{0\dots0}}
				& 253 \\\hline
	\end{tabular}
	\egroup
\end{table}
\par
The second scenario considers the case where $n \ge 12$. By construction, it is evident that a necessary condition for underflow when rounding down and overflow when rounding up occurs when the first $12$ bits take the forms \texttt{\textcolor{sign}{X}\textcolor{direction}{0}%
\textcolor{regime}{000}\textcolor{characteristic}{0000000}} and \texttt{\textcolor{sign}{X}\textcolor{direction}{1}%
\textcolor{regime}{111}\textcolor{characteristic}{111111}} respectively. This is because the direction bit and the regime bits being all zeros or all ones directly imply that the regime value is $7$. Consequently, in the overflow scenario, the necessary condition $c = -255$ or $c = 254$ is satisfied.
\par
To establish a sufficient condition for underflow or overflow, we must examine the mantissa bits. For underflow, these bits must be all zeros, and for overflow, they must be all ones, up to and including the most significant rounding bit, which is consistently positioned since the regime value is always $7$. Given that the bits preceding the mantissa have a fixed length of $12$, we need to inspect the $n - 11$ most significant mantissa bits.
When rounding down would result in underflow or rounding up 
would lead to overflow, an output signal is generated to 
indicate the respective condition for later usage in the encoder
by the rounder
\cite[rtl/{\allowbreak}encoder/{\allowbreak}postencoder.vhd, 
lines 33--107]{code}.
\subsection{Characteristic Precursor Determinator}\label{subsec:encoder-characteristic_precursor_determinator}
\label{subsec:characteristic_precursor}
The initial step in the encoding process involves determining the regime and characteristic bits from the given characteristic. To facilitate this, the characteristic is first normalised into a format that allows for the efficient extraction of both the regime and characteristic bits, regardless of the scenario. We refer to this specific format as the \enquote{characteristic precursor}. The characteristic precursor can be derived using the following
\begin{proposition}[characteristic 
precursor]\label{prop:characteristic_precursor}
	Let $n \in \mathbb{N}_1$ and bit string
	$T := (\textcolor{sign}{S},\textcolor{direction}{D},\textcolor{regime}{R},
	\textcolor{characteristic}{C},\textcolor{mantissa}{M})
	\in {\{0,1\}}^n$ as in Definition~\ref{def:takum} with
	$\takum((\textcolor{sign}{S},\textcolor{direction}{D}, 
	\textcolor{regime}{R},
	\textcolor{characteristic}{C},\textcolor{mantissa}{M})) \notin 
	\{0,\mathrm{NaR} 
	\}$, two's complement characteristic $c$ and regime $r$.
	It holds
	\begin{equation}
		\begin{cases}
			-c-1 & \textcolor{direction}{D} = 0\\
			c & \textcolor{direction}{D} = 1
		\end{cases} + 1 =
		\begin{cases}
			2^r + \sum_{i=0}^{r-1}
				\overline{\textcolor{characteristic}{C}}_i 2^i
				& \textcolor{direction}{D} = 0\\
			2^r + \sum_{i=0}^{r-1}
				\textcolor{characteristic}{C}_i 2^i
				& \textcolor{direction}{D} = 1.
		\end{cases}
	\end{equation}
	It should be noted that $-c-1$ can be obtained from $c$ 
	by a simple bitwise complement, assuming $c$ is 
	represented as a two's complement integer.
\end{proposition}
\begin{proof}
	Let $\textcolor{direction}{D} = 0$.
	Negating any unsigned $m$-bit integer $\ell 
	\in \mathbb{N}_0$ yields $2^m - 1 - \ell$ (used as 
	statement \enquote{a} in the subsequent derivation).
	It follows
	\begin{align}
		(-c - 1) + 1 &=
			-c\\
		&= -\left( -2^{r+1} +1 + \sum_{i=0}^{r-1}
			\textcolor{characteristic}{C}_i 2^i \right)\\
		&= 2^{r+1} - 1 - \sum_{i=0}^{r-1}
			\textcolor{characteristic}{C}_i 2^i\\
		&\overset{a}{=} 2^{r+1} - 1 - \left(
				2^r - 1 - \sum_{i=0}^{r-1}
				\overline{\textcolor{characteristic}{C}}_i 2^i
			\right)\\
		&= 2^r + \sum_{i=0}^{r-1}
			\overline{\textcolor{characteristic}{C}}_i 2^i.
	\end{align}
	Let $\textcolor{direction}{D} = 1$. Then it holds
	\begin{equation}
		c + 1 = 2^r - 1 + \left( \sum_{i=0}^{r-1}
		\textcolor{characteristic}{C}_i 2^i \right) + 1
		= 2^r + \sum_{i=0}^{r-1}
		\textcolor{characteristic}{C}_i 2^i.
	\end{equation}
	As both cases yield the desired results we have proven what was to be 
	shown.
	\qed
\end{proof}
Refer to \cite[rtl/{\allowbreak}encoder/{\allowbreak}postencoder.vhd, lines 109--117]{code} for the VHDL implementation of the characteristic precursor determination. As can be observed, only the 8 least significant bits of the characteristic are negated. This approach stems from the fact that the normalised characteristic, by construction, always has a zero in the most significant bit. Consequently, the subsequent incrementation is performed on a \enquote{standard-form} 8-bit bit string. It is important to note that employing the precursor enables us to circumvent a potentially necessary decrementation (see Definition~\ref{def:takum}), while the incrementation can be efficiently implemented using only half-adders.
\par
However, it is crucial to bear in mind that the characteristic bits will need 
to be negated at a later stage. This step cannot be undertaken at the current 
moment, as the regime value $r$ has not yet been extracted.
\subsection{8-Bit Leading One Detector (LOD)}\label{subsec:encoder-lzc}
By construction, the characteristic precursor, represented as 
an 8-bit unsigned integer, has its most significant set bit 
always at offset $r$ (see 
Proposition~\ref{prop:characteristic_precursor}), where $r$ 
denotes the regime. To determine the regime, it is necessary to 
employ an 8-bit leading one detector (LOD). We adopt the design 
proposed in \cite{lod-2021}, which divides the 8-bit number 
into two 4-bit segments. Each segment is then processed through 
a lookup table, which yields the offset value of the most 
significant bit set to $1$, namely the desired regime value 
$r$. 
\par
If the most significant four bits are all zero, the result from the lower bits' 
lookup table is returned. Conversely, if any of the highest four bits are set 
to $1$, the result from the higher bits' lookup table is returned, with $4$ 
added to account for the fact that the leading one inherently has an offset of 
$4$ due to its position. This adjustment ensures accurate determination of the 
regime. For a detailed implementation, refer to 
\cite[rtl/{\allowbreak}encoder/{\allowbreak}postencoder.vhd, lines 
119--155]{code}.
\subsection{Extended Takum 
Generator}\label{subsec:encoder-extended_takum_generator}
Given the tapered-precision nature of takums, the number of 
mantissa bits varies; without loss of generality, for $n>12$ it 
lies between $n-12$ and $n-5$. The encoder input for the 
mantissa bits has a fixed size of $n-5$ (and is non-existent 
for $n < 5$). However, a regime value of $7$ implies that the 
actual mantissa bit count in the output is $n-12$, not $n-5$.
\par
Rounding the mantissa bits at this stage is not advisable, 
since rounding also occurs in the non-mantissa bits and may 
introduce carry bits originating from the mantissa rounding that
would need to be preserved for later non-mantissa rounding. 
Instead, we construct an \enquote{extended takum} of precision 
$n+7$, which fully accommodates a mantissa of length $n-5$ even 
in the case where there are $7$ characteristic bits. This 
extended takum can then be correctly rounded to $n$ bits in the 
separate rounder.
\par
The initial step in generating the extended takum involves deriving the regime bits from the regime. According to the definition, this process entails examining the direction bit. As detailed in Section~\ref{subsec:characteristic_precursor}, the characteristic precursor must also be inverted to obtain the coded characteristic bits, which occurs during the direction bit check.
\par
Upon inverting the characteristic precursor, the resulting characteristic bits are contained within the lower 7 bits of the output. Given that the characteristic bits vary in length from 0 to 7, there will be unused bits in the upper $7 - \mathit{regime}$ bits. This is not problematic, as we subsequently combine these bits with the mantissa bits and $7$ zero bits, then shift the entire sequence $\mathit{regime}$ bits to the right. This ensures that the characteristic bits always start at index $n+1$ and are followed by the mantissa bits, as intended.
\par
We discard the highest 7 bits of the shifted characteristic and 
mantissa bits, yielding a bit string of length $n+2$. 
Considering that the sign bit, direction bit, and three regime 
bits together occupy a total of 5 bits, the combination of 
these with the previous bit string forms the $(n+7)$-bit 
extended takum
\cite[rtl/{\allowbreak}encoder/{\allowbreak}postencoder.vhd, 
lines 157--176]{code}.
\par
It is noteworthy that the shifter, regardless of $n$, is constrained by a maximum shift offset of 7 (three control bits). This contrasts sharply with posits, where the tapering is unbounded and a shifter must accommodate nearly the entire width $n$.
\subsection{Rounder}
Having determined whether rounding up or down results in an overflow or underflow (see Section~\ref{subsec:encoder-predictor}), and having obtained the $(n+7)$-bit extended takum (see Section~\ref{subsec:encoder-extended_takum_generator}), we can now proceed to combine these elements to achieve a properly rounded $n$-bit takum.
\par
The initial step involves generating two rounding candidates: one for rounding up and one for rounding down. They are obtained
by discarding the extended takum's 7 least significant bits, yielding the rounding-down candidate, and incrementing it to
yield the rounding-up candidate.
The final result is determined by rounding up if either: (1) 
rounding down would cause an underflow; or
(2) rounding up would not cause an overflow, the most significant rounding bit in the extended takum is set and there
is either no tie or we are tied and the rounding-down candidate 
is odd (as we tie to even). If these conditions are not met, we 
opt for the rounding-down candidate.
\par
This approach is preferred over directly incrementing the takum with the 
rounding bit by the result of the logic expression, as it optimises timing: 
while we can immediately start incrementing and determining whether to round up 
or down as soon as the input arrives, direct incrementing would require waiting 
for the logic expression to evaluate before beginning the computation 
\cite[rtl/{\allowbreak}encoder/{\allowbreak}postencoder.vhd, lines 
178--190]{code}.
\subsection{Post-, Logarithmic and Linear Encoder}
\begin{figure}[tbp]
	\begin{center}
		\begin{circuitikz}[scale=0.8, transform shape]
			\node[port, label={[anchor=south west]above right:{$\mathit{sign\_bit}$}}] 
			    (input_sign_bit) at (0,0) {};
			\node[port, label={[anchor=south west]above right:{$\mathit{characteristic}$}}] 
				(input_characteristic) at ([yshift=-1.7cm]input_sign_bit) {};
			\node[port, label={[anchor=south west]above right:{$\mathit{mantissa\_bits}$}}] 
				(input_mantissa_bits) at ([yshift=-1.7cm]input_characteristic) {};
			\node[port, label={[anchor=south west]above right:{$\mathit{is\_zero}$}}] 
				(input_is_zero) at ([yshift=-1.7cm]input_mantissa_bits) {};
			\node[port, label={[anchor=south west]above right:{$\mathit{is\_nar}$}}] 
				(input_is_nar) at ([yshift=-1.7cm]input_is_zero) {};

			\node
			    [draw, rectangle, densely dotted, thick, minimum width=5cm, minimum height=3.5cm,
			    label={[anchor=south east]south east:{E1}}]
			    (E1) at ([xshift=4.1cm,yshift=-6.05cm]input_is_nar.center) {};
			\node
			    [draw, rectangle, densely dotted, thick, minimum width=3cm, minimum height=3cm,
			    label={[anchor=south east]south east:{E2}}]
			    (E2) at ([xshift=6.2cm,yshift=-0.8cm]input_sign_bit.center) {};
			\node
			    [draw, rectangle, densely dotted, thick, minimum width=5cm, minimum height=3cm,
			    label={[anchor=south east]south east:{E4}}]
			    (E4) at ([xshift=7.3cm,yshift=1.05cm]E1.center) {};
			\node
			    [draw, rectangle, densely dotted, thick, minimum width=6cm, minimum height=4.8cm,
			    label={[anchor=south east]south east:{E3}}]
			    (E3) at ([xshift=-1.7cm,yshift=5.45cm]E4.center) {};
			\node
			    [draw, rectangle, densely dotted, thick, minimum width=3.5cm, minimum height=3cm,
			    label={[anchor=south east]south east:{E5}}]
			    (E5) at ([xshift=5.3cm,yshift=0cm]E2.center) {};

			\node[port] 
				(E1_in_characteristic) at ([yshift=-0.8cm]E1.north west) {};
			\node[port] 
				(E1_in_mantissa_bits) at ([yshift=-1.6cm]E1.north west) {};
			\node[port,label={[anchor=south west,align=left,xshift=-0.05cm]above right:{$\mathit{round\_up\_}$\\$\mathit{overflows}$}}] 
				(E1_out_round_up_overflows) at ([yshift=-0.8cm]E1.north east) {};
			\node[port,label={[anchor=north west,align=left,xshift=-0.05cm]below right:{$\mathit{round\_down\_}$\\$\mathit{underflows}$}}] 
				(E1_out_round_down_underflows) at ([yshift=-1.6cm]E1.north east) {};

			\node[port] 
				(E2_in_direction_bit) at ([yshift=-0.8cm]E2.north west) {};
			\node[port] 
				(E2_in_characteristic) at ([yshift=-1.6cm]E2.north west) {};
			\node[port] 
				(E2_out_characteristic_precursor) at ([yshift=-0.8cm]E2.north east) {};

			\node[port] 
				(E3_in_sign_bit) at ([yshift=-0.6cm]E3.north west) {};
			\node[port] 
				(E3_in_direction_bit) at ([yshift=-1.2cm]E3.north west) {};
			\node[port] 
				(E3_in_regime) at ([yshift=-1.8cm]E3.north west) {};
			\node[port] 
				(E3_in_characteristic_precursor) at ([yshift=-2.4cm]E3.north west) {};
			\node[port] 
				(E3_in_mantissa_bits) at ([yshift=-4.0cm]E3.north west) {};
			\node[port] 
				(E3_out_extended_takum) at ([yshift=-0.8cm]E3.north east) {};
			
			\node[port]
				(E4_in_extended_takum) at ([yshift=-0.8cm]E4.north west) {};
			\node[port] 
				(E4_in_round_up_overflows) at ([yshift=-1.6cm]E4.north west) {};
			\node[port] 
				(E4_in_round_down_underflows) at ([yshift=-2.4cm]E4.north west) {};
			\node[port] 
				(E4_out_takum_rounded) at ([yshift=-0.8cm]E4.north east) {};
			
			\node[port]
				(E5_in_takum_rounded) at ([yshift=-0.8cm]E5.north west) {};
			\node[port] 
				(E5_in_is_zero) at ([yshift=-1.6cm]E5.north west) {};
			\node[port] 
				(E5_in_is_nar) at ([yshift=-2.4cm]E5.north west) {};
			\node[port,label={[xshift=-0.1cm]above right:{$\mathit{takum}$}}] 
				(E5_out_takum) at ([yshift=-0.8cm]E5.north east) {};

			\draw (input_sign_bit.center) --
				++(3.2cm, 0)
				node(sign_bit_pickup_line) {} |-
				(E3_in_sign_bit.center);
			\draw[very thick] (input_characteristic.center) --
				++(2.8cm, 0)
				node(characteristic_pickup_line) {} --
				++(0,-7.5cm) --
				++(-2.4cm,0) |-
				(E1_in_characteristic)
				node[midway,fill=white,inner sep=1pt,xshift=0.4cm]{$9$};
			\draw[very thick] (input_mantissa_bits.center) --
				++(2.4cm, 0)
				node(mantissa_pickup_line) {} --
				++(0, -5.5cm) --
				++(-2.4cm, 0cm) |-
				(E1_in_mantissa_bits)
				node[midway,fill=white,inner sep=1pt,xshift=0.55cm]{$n\!-\!5$};
			\draw (input_is_zero.center) --
				++(2.0cm, 0) --
				++(0, -4.4cm) -|
				([xshift=0.8cm,yshift=1.6cm]E3_out_extended_takum.center) -|
				([xshift=-0.8cm]E5_in_is_zero.center) --
				(E5_in_is_zero.center);
			\draw (input_is_nar.center) --
				++(1.6cm, 0) --
				++(0, -3.1cm) -|
				([xshift=1.2cm,yshift=2.0cm]E3_out_extended_takum.center) -|
				([xshift=-0.4cm]E5_in_is_nar.center) --
				(E5_in_is_nar.center);

			\draw (E1_out_round_up_overflows.center) -- (E4_in_round_up_overflows.center);
			\draw (E1_out_round_down_underflows.center) -- (E4_in_round_down_underflows.center);

			\draw[very thick] (characteristic_pickup_line.center) --
				++(0.8cm,0) |-
				(E2_in_characteristic.center);

			\draw[very thick] (mantissa_pickup_line.center |- E3_in_mantissa_bits.center) --
				node[midway,fill=white,inner sep=1pt]{$n\!-\!5$}
				(E3_in_mantissa_bits.center);
			
			\draw[very thick]
				(characteristic_pickup_line.center |- E3_in_direction_bit) --
				++(0.6cm,0)
				node[midway,fill=white,inner sep=1pt]{$9$}
				node(direction_bit_split_top) {};
			\draw[densely dashed] (direction_bit_split_top.center) --
				++(0,-0.8cm)
				node(direction_bit_split_bot) {};
			\draw[very thick]
				(direction_bit_split_bot.center) --
				node[midway,fill=white,inner sep=1pt]{$8$}
				++(0.6cm,0);
			\draw
				(direction_bit_split_top.center) --
				++(0.4cm,0)
				node(direction_bit_inv_ansatz) {};
			\draw node[ieeestd not port, no input leads, no output leads]
				at ([xshift=0.15cm]direction_bit_inv_ansatz.center) (direction_bit_not) {};
			\draw ([xshift=-0.05cm]direction_bit_not.out) -- ++(0.2cm,0)
				node[port, label={[anchor=north,align=left,xshift=0.05cm,yshift=-0.05cm]below:{$\mathit{direction\_bit}$}}] (direction_bit_ansatz) {};

			\draw (direction_bit_ansatz.center) -- (E3_in_direction_bit.center);
			\draw (direction_bit_ansatz.center) |- (E2_in_direction_bit.center);
			
			\draw[very thick]
				(E2_out_characteristic_precursor.center) -|
				++(0.4cm,-3.05cm) -|
				node[midway,fill=white,inner sep=1pt,xshift=1.3cm]{$8$}
				([xshift=1.1cm]direction_bit_ansatz.center)
				node(helper1) {} |-
				(helper1.center |- E3_in_characteristic_precursor.center) --
				++(-2.05cm,0) --
				++(0,-0.5cm)
				node (regime_split) {};
			\draw[very thick]
				(regime_split.center) --
				++(0.5cm,0)
				node(lod8_ansatz) {};
			\node[draw, scale=0.6, fill=white, minimum width=1.5cm, minimum height=1cm,
			      rounded corners=3pt, thick] (lod8)
				at ([xshift=0.45cm]lod8_ansatz.center) {LOD8};
			\draw[very thick]
				(lod8.east) --
				node[midway,fill=white,inner sep=1pt]{$3$}
				++(1.1cm,0) |-
				(E3_in_regime);
			\draw[very thick]
				(helper1.center |- E3_in_characteristic_precursor.center) --
				(E3_in_characteristic_precursor.center);

			\draw[very thick] (E3_out_extended_takum.center) -|
				++(0.4cm,-4.4cm) -|
				([xshift=-0.4cm]E4_in_extended_takum.center) --
				(E4_in_extended_takum.center);

			\draw[very thick]
				(E4_out_takum_rounded.center) --
				++(0.4cm,0) |-
				([yshift=1.2cm]E3_out_extended_takum.center) -|
				node[midway,fill=white,inner sep=1pt,xshift=2.5cm]{$n$}
				([xshift=-1.2cm]E5_in_takum_rounded.center) --
				(E5_in_takum_rounded.center);

			\draw[very thick]
				(E5_out_takum) --
				node[midway,fill=white,inner sep=1pt]{$n$}
				++(1.0cm,0)
				node[port] {};

			\draw[very thick] (E1_in_characteristic) -|
				++(0.5cm,0.4cm) -- ++(1.5cm,0)
				node(E1cmp1ansatz) {};
			\draw[very thick]  ([xshift=0.5cm]E1_in_characteristic.center) |-
				++(1.5cm,-0.5cm)
				node(E1cmp2ansatz) {};
			\node[draw, scale=0.6, fill=white, minimum width=1.5cm, minimum height=1cm,
				rounded corners=3pt, thick] (E1cmp1)
				at ([xshift=0.45cm,yshift=-0.15cm]E1cmp1ansatz.center) {CMPEQ};
			\node[draw, scale=0.6, fill=white, minimum width=1.5cm, minimum height=1cm,
				rounded corners=3pt, thick] (E1cmp2)
				at ([xshift=0.45cm,yshift=-0.15cm]E1cmp2ansatz.center) {CMPEQ};
			\draw[very thick] ([yshift=-0.15cm]E1cmp1.west) --
				node[midway,fill=white,inner sep=1pt,xshift=0.04cm]{$9$}
				++(-0.5cm,0)
				node[port,label={[label distance=-0.1cm]left:{$254$}}] {};
			\draw[very thick] ([yshift=-0.15cm]E1cmp2.west) --
				node[midway,fill=white,inner sep=1pt,xshift=0.04cm]{$9$}
				++(-0.5cm,0)
				node[port,label={[label distance=-0.1cm]left:{$-255$}}] {};

			\draw[very thick] (E1_in_mantissa_bits) --
				++(0.5cm,0) node(E1split) {};
			\draw[densely dashed]
				(E1split.center) --
				++(0,-1.4cm) node(E1split_bot){};
			\node(E1split_top) at ([yshift=-0.8cm]E1split.center) {};
			\draw[very thick] (E1split_bot.center) -- ++(0.8,0)
				node[midway,fill=white,inner sep=1pt]{$6$}
				node[port] {};
			\draw[very thick] (E1split_top.center) -- ++(1.3,0)
				node[midway,fill=white,inner sep=1pt]{$n\!-\!11$}
				node[port](E1aftersplit) {};

			\node[and port,no input leads,no output leads] (E1and)
				at ([xshift=0.9cm]E1aftersplit.center) {};
			\node[nor port,no input leads,no output leads] (E1nor)
				at ([xshift=0.85cm,yshift=-0.6cm]E1aftersplit.center) {};
			\draw[very thick] (E1aftersplit.center) --
				([xshift=0.1cm]E1and.west);
			\draw[very thick] ([xshift=0.2cm]E1aftersplit.center) |-
				([xshift=0.1cm]E1nor.west);

			\node[and port] (E1and1) at ([xshift=-0.6cm]E1_out_round_up_overflows) {};
			\node[and port] (E1and2) at ([xshift=-0.6cm]E1_out_round_down_underflows) {};
			\draw (E1and1.out) -- (E1_out_round_up_overflows |- E1and1.out);
			\draw (E1and2.out) -- (E1_out_round_up_overflows |- E1and2.out);
			\draw (E1cmp1.east) -| (E1and1.in 1);
			\draw (E1cmp2.east) -| (E1and2.in 1);
			\draw ([xshift=-0.1cm]E1and.out) -| ([xshift=-0.8cm]E1and1.in 2) --
				(E1and1.in 2);
			\draw ([xshift=-0.1cm]E1nor.out) -| ([xshift=-0.4cm]E1and2.in 2) --
				(E1and2.in 2);

			\draw[very thick] (E2_in_characteristic.center) --
				++(0.25cm,0) node(E2split) {};
			\draw[densely dashed]
				([xshift=-0.006cm]E2split.center) --
				++(0,-0.9cm) node(E2split_bot){};
			\node(E2split_top) at ([yshift=-0.3cm]E2split.center) {};
			\draw[very thick] ([xshift=-0.006cm]E2split_bot.center) -- ++(1.0,0)
				node[midway,fill=white,inner sep=1pt]{$8$}
				node[port](E2aftersplit) {};
			\draw (E2split_top.center) -- ++(1.0,0)
				node[port] {};

			\draw ([xshift=1.2cm]E2_in_direction_bit)
				node[scale=0.5,iampshape,fill=white,t={\ \,MUX}] (E2mux) {};
			\draw[very thick] (E2aftersplit.center) -- ++(0.7cm,0)
				node(E2invansatz) {};
			\draw[very thick] ([xshift=0.4cm]E2aftersplit.center) |-
				++(-1.0cm,0.9cm) |- ([yshift=-0.15cm]E2mux.west);
			\draw node[ieeestd not port, no input leads, no output leads] at ([xshift=0.15cm]E2invansatz.center) (E2not) {};
			\draw[very thick] ([xshift=-0.06cm]E2not.out) -| ++(0.3cm,1.15cm)
				-- ++(-2.1cm,0) |- ([yshift=0.15cm]E2mux.west);

			\draw (E2_in_direction_bit.center) -| ++(0.244cm,0.55cm) -|
				([yshift=-0.05cm]E2mux.north);

			\draw[very thick] (E2mux.east) -- ++(0.1,0)
				node(E2incansatz) {};
			\node[draw, scale=0.6, fill=white, minimum width=1.5cm, minimum height=1cm,
			      rounded corners=3pt, thick] (E2inc) at ([xshift=0.45cm]E2incansatz.center) {INC};
			\draw[very thick]
				(E2inc.east) --
				node[midway,fill=white,inner sep=1pt,xshift=-0.02cm]{$8$}
				(E2_out_characteristic_precursor.center);

			\draw[very thick] (E3_out_extended_takum.center) --
				node[midway,fill=white,inner sep=1pt]{$n\!+\!7$}
				++(-1.2cm,0)
				node(E3_fusion_top) {};
			\draw[densely dashed] ([yshift=0.2cm]E3_fusion_top.center) --
				++(0,-3.85cm)
				node(E3_fusion_bot) {};
			\draw (E3_in_sign_bit.center) --
				(E3_fusion_top.center |- E3_in_sign_bit.center);
			\draw (E3_in_direction_bit.center) --
				(E3_fusion_top.center |- E3_in_direction_bit.center);
			
			\draw[very thick] (E3_in_regime.center) -- ++(0.5cm,0)
				node(E3_regime_split) {};
			\draw[very thick] (E3_regime_split.center) |- ++(0.2cm,0.15cm)
				node(E3_regime_inv_ansatz) {};
			\draw node[ieeestd not port, no input leads, no output leads]
				at ([xshift=0.15cm]E3_regime_inv_ansatz.center) (E3not1) {};
			\draw[very thick] ([xshift=-0.06cm]E3not1.out) -- ++(0.2cm,0)
				node(E3_mux1tip) {};
			\draw[very thick] (E3_regime_split.center) |- ++(0.75cm,-0.15cm);
			\draw ([xshift=0.25cm,yshift=-0.15cm]E3_mux1tip)
				node[scale=0.5,iampshape,fill=white,t={\ \,MUX}] (E3mux1) {};
			\draw ([yshift=-0.05cm]E3mux1.north) --
				(E3mux1.north |- E3_in_direction_bit.center);
			\draw[very thick] (E3mux1.east) --
				node[midway,fill=white,inner sep=1pt]{$3$}
				(E3_fusion_top |- E3mux1.east);

			\draw[very thick] (E3_in_characteristic_precursor.center) --
				++(0.5cm,0) node(E3split) {};
			\draw[densely dashed]
				([xshift=-0.006cm]E3split.center) --
				++(0,-0.6cm) node(E3split_bot){};
			\node(E3split_top) at ([yshift=-0.2cm]E3split.center) {};
			\draw[very thick] ([xshift=-0.006cm]E3split_bot.center) -- ++(0.6,0)
				node[midway,fill=white,inner sep=1pt]{$7$}
				node[port](E3aftersplit) {};
			\draw (E3split_top.center) -- ++(0.6,0)
				node[port] {};

			\draw[very thick] (E3aftersplit.center) -- ++(0.2cm,0)
				node(E3_characteristic_precursor_split) {};
			\draw[very thick] (E3_characteristic_precursor_split.center) |- ++(0.2cm,0.15cm)
				node(E3_characteristic_precursor_inv_ansatz) {};
			\draw node[ieeestd not port, no input leads, no output leads]
				at ([xshift=0.15cm]E3_characteristic_precursor_inv_ansatz.center) (E3not2) {};
			\draw[very thick] ([xshift=-0.06cm]E3not2.out) -- ++(0.2cm,0)
				node(E3_mux2tip) {};;
			\draw[very thick] (E3_characteristic_precursor_split.center) |- ++(0.75cm,-0.15cm);
			\draw ([xshift=0.25cm,yshift=-0.15cm]E3_mux2tip)
				node[scale=0.5,iampshape,fill=white,t={\ \,MUX}] (E3mux2) {};
			\draw ([yshift=-0.05cm]E3mux2.north) --
				(E3mux2.north |- E3_in_direction_bit.center);
			\draw[very thick] (E3mux2.east) -| ++(0.2cm,-0.5cm) --
				node[midway,fill=white,inner sep=1pt]{$7$}
				++(-2.25cm,0.0) |- ++(0.9cm,-0.3cm)
				node(E3_split2_top) {};
			\draw[densely dashed](E3_split2_top.center) -- ++(0,-0.6cm)
				node(E3_split2_bot) {};
			\draw[very thick] (E3_in_mantissa_bits.center) --
				(E3_split2_top |- E3_in_mantissa_bits.center);
			\draw[very thick] (E3_split2_bot.center) --
				node[midway,fill=white,inner sep=1pt]{$7$}
				++(-0.6cm,0)
				node[port,label={[label distance=-0.1cm]left:{$0$}}]{};
			\draw[very thick] ([yshift=0.2cm]E3_split2_bot.center) --
				node[midway,fill=white,inner sep=1pt]{$n\!+\!9$}
				++(1.1cm,0)
				node(E3_shiftr_ansatz) {};
			\node[draw, scale=0.6, fill=white, minimum width=1.5cm, minimum height=1cm,
			      rounded corners=3pt, thick] (E3_shiftr)
				at ([xshift=0.45cm]E3_shiftr_ansatz.center) {SHIFTR};
			\draw[very thick] ([yshift=-0.15cm]E3_regime_split.center) --
				++(0,-0.3cm)
				node(E3_shiftr_regime_input_start) {};
			\draw[very thick] ([yshift=0.025cm]E3_shiftr_regime_input_start.center) -|
				(E3_shiftr.north);

			\draw[very thick] (E3_fusion_bot.center) --
				node[midway,fill=white,inner sep=1pt]{$n\!+\!2$}
				++(-1.2cm,0)
				node(E3_fusion2_top) {};
			\draw[densely dashed] (E3_fusion2_top.center) --
				++(0,0.9cm)
				node(E3_fusion2_bot) {};
			\draw[very thick] (E3_fusion2_bot.center) --
				node[midway,fill=white,inner sep=1pt]{$7$}
				++(0.6cm,0) node[port]{};
			\draw[very thick] (E3_shiftr.east) --
				(E3_fusion2_bot.center |- E3_shiftr.east);
			
			\draw ([xshift=-0.5cm]E4_out_takum_rounded) node[scale=0.5,iampshape,fill=white,t={\ \,MUX}] (E4_mux1) {};
			\draw[very thick] (E4_mux1.east) -- (E4_out_takum_rounded.center);

			\draw[very thick] (E4_in_extended_takum.center) -- ++(0.8cm,0)
				node(E4_split1_ansatz) {};
			\draw[densely dashed] ([yshift=0.4cm]E4_split1_ansatz.center)
				node(E4_split1_top) {} --
				++(0,-1.5cm)
				node(E4_split1_bot) {};
			\draw[very thick] (E4_split1_top.center) --
				node[midway,fill=white,inner sep=1pt]{$n\!-\!1$}
				++(1.0cm,0)
				node(E4_fusion_top){};

			\draw ([yshift=0.8cm]E4_split1_bot.center) -- ++(1.0cm,0)
				node(E4_fusion_bot) {};
			\draw ([yshift=0.4cm]E4_split1_bot.center) -- ++(0.8cm,0)
				node(E4_takum6) {};

			\draw[densely dashed] (E4_fusion_top.center) -- 	
				(E4_fusion_bot.center);
			\draw[very thick] ([yshift=-0.2cm]E4_fusion_top.center) --
				node[midway,fill=white,inner sep=1pt]{$n$}
				++(0.6cm,0)
				node(E4_takum_inc_split) {};
			\draw[very thick] ([yshift=0.025cm]E4_takum_inc_split.center) |-
				([yshift=-0.15cm]E4_mux1.west);
			\draw[very thick] (E4_takum_inc_split.center) |- ++(0.3cm,0.1cm)
				node(E4_takum_inc_ansatz) {};
			\node[draw, scale=0.6, fill=white, minimum width=1.5cm, minimum height=1cm,
			      rounded corners=3pt, thick] (E4inc) at ([xshift=0.45cm]E4_takum_inc_ansatz.center) {INC};
			\draw[very thick] (E4inc.east) --
				++(0.3cm,0) |-
				([yshift=0.15cm]E4_mux1.west);

			\draw[very thick] (E4_split1_bot.center) --
				node[midway,fill=white,inner sep=1pt]{$6$}
				++(0.8cm,0)
				node(E4_nor_ansatz) {};
			\node[or port,no input leads,no output leads] (E4nor)
				at ([xshift=0.2cm]E4_nor_ansatz) {};

			\node[or port,no input leads,no output leads] (E4ormain)
				at ([xshift=2.4cm]E4nor.center) {};
			\draw ([xshift=-0.1cm]E4ormain.out) -| ([yshift=0.05cm]E4_mux1.south);
			\draw (E4_in_round_down_underflows) -| ++(0.4cm,-0.3cm) --
				++(3.43cm,0) |- ([xshift=0.1cm]E4ormain.in 2);

			\node[and port,no input leads,no output leads] (E4andmain)
				at ([xshift=-0.7cm,yshift=0.5cm]E4ormain.center) {};
			\draw ([xshift=-0.1cm]E4andmain.out) --
				++(0.1cm,0) |- ([xshift=0.1cm]E4ormain.in 1);

			\draw (E4_takum6.center) |-
				([xshift=0.1cm,yshift=0.1cm]E4andmain.west);

			\draw (E4_in_round_up_overflows.center) -|++(0.6cm,-0.7cm) --
				++(0.2cm,0) node(E4_inv_ansatz) {};
			\draw node[ieeestd not port, no input leads, no output leads]
				at ([xshift=0.15cm]E4_inv_ansatz.center) (E4not) {};
			\draw ([xshift=-0.05cm]E4not.out) -- ++(1.95cm,0) |- ([xshift=0.1cm,yshift=-0.1cm]E4andmain.west);

			\node[or port,no input leads,no output leads] (E4or)
				at ([xshift=0.8cm,yshift=0.085cm]E4nor.center) {};
			\draw (E4_fusion_bot.center) -- ++(0.3cm,0) |-
				([xshift=0.1cm]E4or.in 1);
			\draw ([xshift=-0.1cm]E4nor.out) -- ++(0.1cm,0) |-
				([xshift=0.1cm]E4or.in 2);

			\draw ([xshift=-0.1cm]E4or.out) -- ++(0.1cm,0) |-
				([xshift=0.1cm]E4andmain.west);

			\draw ([xshift=-0.8cm]E5_out_takum) node[scale=0.5,iampshape,fill=white,t={\ \,MUX}] (E5_mux1) {};
			\draw[very thick] (E5_mux1.east) -- (E5_out_takum.center);

			\node[or port,no input leads,no output leads] (E5or)
				at ([yshift=-1.515cm]E5_mux1.center) {};
			\draw (E5_in_is_nar.center |- E5or.in 2) -- ([xshift=0.1cm]E5or.in 2);
			\draw ([xshift=-0.1cm]E5or.out) -|
				++(0.3cm,2.0cm) -|
				([yshift=-0.05cm]E5_mux1.north);
			\draw (E5_in_is_zero.center) -- ++(0.4cm,0) |-
				([xshift=0.1cm]E5or.in 1);

			\draw[very thick] (E5_in_takum_rounded.center) --
				++(0.4cm,0) |- ([yshift=0.15cm]E5_mux1.west);

			\draw[densely dashed]
				([xshift=2.0cm,yshift=-0.15cm]E5_in_takum_rounded.center)
				node(E5_fusion_top) {} --
				++(0,-1.0cm)
				node(E5_fusion_bot) {};
			\draw (E5_fusion_top.center) -|
				([xshift=0.6cm]E5_in_is_nar.center);
			\draw[very thick] (E5_fusion_bot.center) --
				node[midway,fill=white,inner sep=1pt]{$n\!-\!1$}
				++(-1.0cm,0)
				node[port,label={[label distance=-0.1cm]left:{$0$}}] {};

			\draw[very thick] (E5_fusion_top.center) --
				node[midway,fill=white,inner sep=1pt]{$n$}
				([yshift=-0.15cm]E5_mux1.west);
		\end{circuitikz}
	\end{center}
	\caption{
		The logic circuit of the postencoder, largely separated into five main entities:
		the underflow/overflow predictor (E1), the characteristic precursor determinator (E2),
		the extended takum generator (E3), the rounder (E4) and the output driver (E5).
		We assume $n \ge 12$ (thus omitting special 
		case handling in the underflow/overflow 
		predictor for $n < 12$) for simplicity; the
		implemented postencoder works for any $n \ge 
		2$. 
		Vertical dashed lines
		indicate where the strands of a multi-signal are split up or combined.
	}
	\label{fig:schematic-postencoder}
\end{figure}
With all components in place, we can now integrate them into a single comprehensive component: the postencoder.
It is made up of five main entities: the underflow/overflow predictor (E1), the characteristic precursor determinator (E2),
the extended takum generator (E3), the rounder (E4) and the output driver (E5). The determination of the direction bit at the outset is given separately, as it is simply the inverted sign bit of the characteristic \cite[rtl/{\allowbreak}encoder/{\allowbreak}postencoder.vhd]{code}. A full logic circuit is outlined in Figure~\ref{fig:schematic-postencoder}.
\par
For encoding logarithmic takum values, as outlined in 
\cite[rtl/{\allowbreak}encoder/{\allowbreak}encoder\_{\allowbreak}logarithmic.vhd]{code},
 the characteristic and mantissa bits derived from the barred 
logarithmic value, as defined in
\eqref{eq:representation-barred_logarithmic}, are passed into 
the postencoder. Similarly, the linear takum encoder, as 
described in 
\cite[rtl/{\allowbreak}encoder/{\allowbreak}encoder\_{\allowbreak}linear.vhd]{code},
 employs the postencoder following the conversion of the 
internal representation from equation 
(\ref{eq:representation-two_s_complement}). The characteristic 
is computed from the exponent, with conditional negation 
applied based on $\textcolor{sign}{S}$.
\section{Evaluation}\label{sec:evaluation}
Having devised decoders and encoders for both 
logarithmic and linear takums, we 
now proceed to evaluate their performance relative to 
the most effective posit 
codecs. The synthesis was conducted using Vivado 2024.1 
on a Kintex UltraScale+ 
KCU116 Evaluation Platform (FPGA part number 
XCKU5P-2FFVB676E) with a maximum
operating frequency of \SI{725}{\mega\hertz} 
\cite{fpga}. We employed the default synthesis strategy 
(Vivado Synthesis Defaults, 2024).
\par
First, both the encoder and decoder are validated using a testbench (see \cite[{simulation/{\allowbreak}decoder/{\allowbreak}predecoder\_{\allowbreak}tb.vhd, simulation/{\allowbreak}encoder/{\allowbreak}postencoder\_{\allowbreak}tb.vhd}]{code}). While the decoder is verified against a reference implementation, the encoder is assessed through a round-trip test, employing the validated decoder as the initial stage. This procedure ensures the correct operation of both components.
\par
Our takum codec is compared against the sign-magnitude FloPoCo codec (\enquote{FloPoCo-SM}, see \cite{flopoco-1-2020,flopoco-2-2021,flopoco-3-2021}), which utilizes the internal representation detailed in (\ref{eq:representation-floating_point}), and the enhanced two's complement FloPoCo codec (\enquote{FloPoCo-2C}) \cite{flopoco-4-2022}, which employs the optimised internal representation described in (\ref{eq:representation-two_s_complement}) and is integrated into the PERCIVAL RISC-V core project \cite{percival-2022,percival-64_bit-2024}. For a comparison of the internal representations, refer to Section~\ref{sec:internal_representations}.
\par
The PACoGen codec \cite{pacogen-2019} has been excluded from this comparison 
due to \cite[Figure~7]{flopoco-4-2022} indicating that its performance is 
generally inferior or at best comparable to that of the FloPoCo-2C codec. 
Similarly, the MArTo codec \cite{posit-hardware_cost-2019} was not included for 
the same reason. The posit codecs used in this evaluation have been extracted 
and incorporated into \cite[rtl/third\_party/]{code}.
\subsection{Decoder}
We begin by assessing the decoder's performance by measuring 
its maximum latency and look-up table (LUT) consumption on the 
reference FPGA. With respect to maximum latency, as illustrated 
in Figure~\ref{subfig:decoder-latency}, it is noteworthy that 
our results for both FloPoCo-SM and FloPoCo-2C are consistent 
with those reported in \cite{flopoco-4-2022}, despite the use 
of a different FPGA platform.
\par
\begin{figure}[tbp]
	\begin{center}
		\subfloat[maximum latency/ns]{
			\label{subfig:decoder-latency}
			\begin{tikzpicture}
				\begin{axis}[
					scale only axis,
					width=0.40\textwidth,
					height=0.2472\textwidth,
					xtick={8,16,32,64,128},
					xlabel={$n$},
					xlabel near ticks,
					grid=both,
					legend style={nodes={scale=0.7, transform shape}},
					legend style={at={(0.97,0.55)},anchor=east},
					grid style={line width=.1pt, draw=gray!10},
				]
					\addplot [characteristic,thick,mark=*, mark size=1.3pt] table [col sep=comma, row sep=\\] {
						8,3.189\\
						16,3.658\\
						32,3.658\\
						64,3.663\\
					};
					\addlegendentry{Takum}
					\addplot [characteristic,mark=x, mark size=1.3pt] table [col sep=comma, row sep=\\] {
						8,3.058\\
						16,3.652\\
						32,3.652\\
						64,3.663\\
					};
					\addlegendentry{Linear Takum}
					\addplot [sign,densely dashed,mark=*, mark size=1.3pt] table [col sep=comma, row sep=\\] {
						8,3.374\\
						16,3.911\\
						32,4.914\\
						64,5.864\\
					};
					\addlegendentry{FloPoCo-2C}
					\addplot [mantissa,densely dotted,mark=*, mark size=1.3pt] table [col sep=comma, row sep=\\] {
						8,3.190\\
						16,4.761\\
						32,5.535\\
					};
					\addlegendentry{FloPoCo-SM}
				\end{axis}
			\end{tikzpicture}
		}
		\hspace{0.5cm}
		\subfloat[CLB LUT count]{
			\label{subfig:decoder-luts}
			\begin{tikzpicture}
				\begin{axis}[
					scale only axis,
					width=0.40\textwidth,
					height=0.2472\textwidth,
					xtick={8,16,32,64,128},
					xtick={8,16,32,64,128},
					xlabel={$n$},
					xlabel near ticks,
					grid=both,
					ymin=-5,
					ymax=290,
					legend style={nodes={scale=0.7, transform shape}},
					legend style={at={(0.03,0.95)},anchor=north west},
					grid style={line width=.1pt, draw=gray!10},
				]
					\addplot [characteristic,thick,mark=*, mark size=1.3pt] table [col sep=comma, row sep=\\] {
						8,22\\
						16,39\\
						32,68\\
						64,125\\
					};
					\addlegendentry{Takum}
					\addplot [characteristic,mark=x, mark size=1.3pt] table [col sep=comma, row sep=\\] {
						8,21\\
						16,39\\
						32,67\\
						64,125\\
					};
					\addlegendentry{Linear Takum}
					\addplot [sign,densely dashed,mark=*, mark size=1.3pt] table [col sep=comma, row sep=\\] {
						8,15\\
						16,57\\
						32,106\\
						64,250\\
					};
					\addlegendentry{FloPoCo-2C}
					\addplot [mantissa,densely dotted,mark=*, mark size=1.3pt] table [col sep=comma, row sep=\\] {
						8,32\\
						16,62\\
						32,137\\
					};
					\addlegendentry{FloPoCo-SM}
				\end{axis}
			\end{tikzpicture}
		}
	\end{center}
	\caption{
		Evaluation results for the decoder in terms of latency and LUT consumption.
	}
	\label{fig:decoder}
\end{figure}
The results indicate that the takum decoders outperform 
both reference 
decoders, even for $n=8$ where FloPoCo-SM beats 
FloPoCo-2C. Overall, the takum 
decoders' latencies are up to $\SI{38}{\percent}$ lower 
than FloPoCo-2C's.
For decoder widths ranging from 8 to 64 bits, the 
FloPoCo-2C decoders can 
operate at frequencies of up to approximately 
\SI{170}{\mega\hertz}, whereas
the takum decoders achieve up to approximately 
\SI{270}{\mega\hertz}.
\par
A similar trend is observed in the LUT consumption, as shown in 
Figure~\ref{subfig:decoder-luts}. For $n=8$, all 
implementations exhibit comparable LUT usage and nearly linear 
delay growth. However, the takum decoder demonstrates a 
significantly lower CLB LUT consumption, with up to 
$\SI{50}{\percent}$ less usage compared to the most efficient 
posit reference.
\subsection{Encoder}
\begin{figure}[tbp]
	\begin{center}
		\subfloat[maximum latency/ns]{
			\label{subfig:encoder-latency}
			\begin{tikzpicture}
				\begin{axis}[
					scale only axis,
					width=0.40\textwidth,
					height=0.2472\textwidth,
					xtick={8,16,32,64,128},
					xtick={8,16,32,64,128},
					xlabel={$n$},
					xlabel near ticks,
					ymax=5.1,
					grid=both,
					legend style={nodes={scale=0.7, transform shape}},
					legend style={at={(0.03,0.95)},anchor=north west},
					grid style={line width=.1pt, draw=gray!10},
				]
					\addplot [characteristic,thick,mark=*, mark size=1.3pt] table [col sep=comma, row sep=\\] {
						8,4.161\\
						16,4.221\\
						32,4.275\\
						64,4.298\\
					};
					\addlegendentry{Takum}
					\addplot [characteristic,mark=x, mark size=1.3pt] table [col sep=comma, row sep=\\] {
						8,4.122\\
						16,4.180\\
						32,4.234\\
						64,4.254\\
					};
					\addlegendentry{Linear Takum}
					\addplot [sign,densely dashed,mark=*, mark size=1.3pt] table [col sep=comma, row sep=\\] {
						8,3.829\\
						16,4.349\\
						32,4.449\\
						64,4.908\\
					};
					\addlegendentry{FloPoCo-2C}
				\end{axis}
			\end{tikzpicture}
		}
		\hspace{0.5cm}
		\subfloat[CLB LUT count]{
			\label{subfig:encoder-luts}
			\begin{tikzpicture}
				\begin{axis}[
					scale only axis,
					width=0.40\textwidth,
					height=0.2472\textwidth,
					xtick={8,16,32,64,128},
					xtick={8,16,32,64,128},
					xlabel={$n$},
					xlabel near ticks,
					grid=both,
					legend style={nodes={scale=0.7, transform shape}},
					legend style={at={(0.03,0.95)},anchor=north west},
					grid style={line width=.1pt, draw=gray!10},
				]
					\addplot [characteristic,thick,mark=*, mark size=1.3pt] table [col sep=comma, row sep=\\] {
						8,39\\
						16,71\\
						32,140\\
						64,221\\
					};
					\addlegendentry{Takum}
					\addplot [characteristic,mark=x, mark size=1.3pt] table [col sep=comma, row sep=\\] {
						8,40\\
						16,66\\
						32,140\\
						64,237\\
					};
					\addlegendentry{Linear Takum}
					\addplot [sign,densely dashed,mark=*, mark size=1.3pt] table [col sep=comma, row sep=\\] {
						8,24\\
						16,56\\
						32,128\\
						64,269\\
					};
					\addlegendentry{FloPoCo-2C}
				\end{axis}
			\end{tikzpicture}
		}
	\end{center}
	\caption{
		Evaluation results for the encoder in terms of latency and LUT consumption.
	}
	\label{fig:encoder}
\end{figure}
In our comparison of encoders it is important to note that 
FloPoCo-SM is not included. This omission is due to the fact 
that FloPoCo-SM lacks a distinct posit encoder within its 
codebase that could be evaluated separately. Instead, the posit 
encoder is integrated with the arithmetic logic in FloPoCo-SM.
This poses no issue, as FloPoCo-2C, its superior successor, 
provides distinct decoders and encoders.
\par
Additionally, while the takum encoder incorporates
comprehensive rounding logic, the FloPoCo-2C encoder does not 
and depends on external rounding information provided by the 
caller, which must be determined beforehand and limits the 
posit encoder's usability. Within the scope of this comparison, 
the proposed takum encoder is consequently at a significant 
disadvantage, as it has to derive this rounding information on 
its own.
\par
Moreover, for $n=16$, the synthesiser defaults to using LUTs instead of a CARRY8 chain for computing $\mathit{takum\_rounded\_up}$, increasing the delay. To address this, synthesis settings were adjusted to enforce the use of a CARRY8 chain. The same approach was taken
with the linear and logarithmic takum encoders to enforce consistent synthesis for $n=8$.
\par
An examination of the maximum latency presented in 
Figure~\ref{subfig:encoder-latency} 
reveals that, although the takum encoders exhibit 
higher latency for $n=8$, 
their latencies increase only marginally as $n$ grows, 
ultimately being up to 
$\SI{13}{\percent}$ lower than those of the FloPoCo-2C 
encoder. In contrast, 
the latency of the (non-rounding) FloPoCo-2C encoder 
increases more 
significantly with larger values of $n$. Both the 
logarithmic and linear 
variants of the takum encoder display comparable 
performance characteristics 
across the range of $n$.
For encoder widths ranging from 8 to 64 bits, the 
FloPoCo-2C encoders can 
operate at frequencies of up to approximately 
\SI{200}{\mega\hertz}, whereas the
takum encoders achieve up to approximately 
\SI{230}{\mega\hertz}.
\par
In terms of slice LUT consumption, as illustrated in 
Figure~\ref{subfig:encoder-luts}, the takum and posit encoders 
show comparable usage. The non-rounding FloPoCo-2C encoder 
exhibits a slight advantage for $n \le 32$, whereas the takum 
encoders have an advantage for $n=64$.
\section{Conclusion and Outlook}\label{sec:conclusion_and_outlook}
We have demonstrated the design and implementation of a takum hardware codec using VHDL,
proposing a new internal LNS representation for improved 
performance.
Our evaluations reveal that the takum codec significantly 
outperforms the leading posit codec, FloPoCo-2C, 
while not relying on outside rounding information at the 
encoding stage, and demonstrates near-optimal scalability with 
respect to $n$.
This superior performance can largely be attributed to the fact 
that 
the coded exponent in takums is bounded, whereas for posits it 
is unbounded, and the use of elegant transforms made possible 
by it. In contrast, posits require logic that encompasses the 
entire width of the binary string, spanning all bits during 
both decoding and encoding operations (for instance, shifts), 
rather than just the initial 12 bits.
Therefore, takums can be said to be generally more 
hardware-efficient compared to posits on FPGA platforms, at 
least in regard to the codec design, since the full design 
of an arithmetic processing unit (APU) introduces additional 
variables. Nevertheless, given that posits and takums share the 
same internal representation, the codec largely remains the 
sole distinguishing component.
\par
Future work will evaluate the performance of the takum 
codec on VLSI systems, 
given that the complexity of FPGA implementations does 
not always correlate 
with VLSI complexity. Another aspect concerns the evaluation of 
the codec implemented within a full APU, potentially as a 
pipelined design, in direct comparison with the current state 
of the art in IEEE 754 APUs. An additional task would be 
discussing the implementation of the \enquote{quire} 
accumulator, as used in posit arithmetic.
Furthermore, it is warranted to investigate the effects 
of the novel choice of 
base $\sqrt{e}$ in (logarithmic) takums on the hardware 
implementation of an 
arithmetic core. This aspect is tentatively addressed 
in  
\cite[Section~4.4]{2024-takum}. However, a 
comprehensive hardware 
implementation is necessary for a thorough assessment.
\begingroup
\sloppy
\printbibliography
\endgroup
\end{document}